\documentclass{birkjour}

\usepackage{amsmath,amssymb,amsthm,eucal}
\usepackage{tikz}
\usetikzlibrary{calc,matrix,patterns}

\newcommand{\be}{\begin{equation}}   
\newcommand{\ee}{\end{equation}}

\newcommand{\Dl}{\Delta}          
\newcommand{\dl}{\delta}          
\newcommand{\eps}{\varepsilon}    
\newcommand{\Ga}{\Gamma}          
\newcommand{\om}{\omega}          
\renewcommand{\th}{\theta}        
\newcommand{\vf}{\varphi}         
\newcommand{\bzeta}{{\boldsymbol{\zeta}}}

\newcommand{\bC}{\mathbb{C}}      
\newcommand{\bM}{\mathbb{M}}      
\newcommand{\bN}{\mathbb{N}}      
\newcommand{\bR}{\mathbb{R}}      
\newcommand{\bZ}{\mathbb{Z}}      
\newcommand{\pp}{\mathrm{pp}}		

\newcommand{\MS}{\mathrm{MS}}	  
\newcommand{\sD}{\mathcal{D}}     
\newcommand{\sF}{\mathcal{F}}     
\newcommand{\sP}{\mathcal{P}}     

\newcommand{\bal}{\mathrm{bal}}   
\newcommand{\homog}{\mathrm{hom}} 

\newcommand{\del}{\partial}       
\newcommand{\downto}{\downarrow}  
\newcommand{\less}{\setminus}     
\newcommand{\ovl}{\overline}      
\newcommand{\ox}{\otimes}         
\newcommand{\unl}{\underline}     
\newcommand{\x}{\times}           
\renewcommand{\:}{\colon}         

\DeclareMathOperator{\sd}{sd}      

\newcommand{\set}[1]{\{\,#1\,\}}    
\newcommand{\word}[1]{\quad\mbox{#1}\quad} 
\def\wick:#1:{\mathopen:#1\mathclose:} 

\newcommand{\fd}[2]{\frac{\delta#1}{\delta\varphi(#2)}} 
\def\duo<#1,#2>{\langle#1,#2\rangle} 

\renewcommand{\labelenumi}{\textup{(\roman{enumi})}}

\newtheorem{thm}{Theorem}[section]
 
 \newtheorem{lem}[thm]{Lemma}
 \newtheorem{prop}[thm]{Proposition}
 \theoremstyle{definition}
 \newtheorem{defn}[thm]{Definition}
 \theoremstyle{remark}
 \newtheorem{rem}[thm]{Remark}
 \newtheorem*{ex}{Example}
 \numberwithin{equation}{section}

\tikzset{
photon/.style={decorate, decoration={snake,amplitude=4pt, segment length=5pt}, draw=red},
particle/.style={draw=blue, postaction={decorate}, decoration={markings,mark=at position .5 with {\arrow[draw=blue]{>}}}},
antiparticle/.style={draw=blue, postaction={decorate}, decoration={markings,mark=at position .5 with {\arrow[draw=blue]{<}}}},
gluon/.style={decorate, draw=green, decoration={coil,amplitude=4pt, segment length=5pt}}
}

\begin{document}

\title[Scaling and mass expansion]
{The scaling and mass expansion}

\author[M.~D\"utsch]{Michael D\"utsch}

\address{%
Max Planck Institute for Mathematics in the Sciences\\
Inselstrasse 22\\
D-04103 Leipzig\\
Germany}

\email{michael.duetsch@theorie.physik.uni-goettingen.de}




\begin{abstract}
The scaling and mass expansion (shortly 'sm-expansion') is a new
axiom for causal perturbation theory, which is a stronger version of
a frequently used renormalization condition in terms of
Steinmann's scaling degree \cite{EpsteinG73, BrunettiF00}.

If one quantizes the underlying free theory by using a Hadamard function
(which is smooth in $m\geq 0$), 
one can reduce renormalization of a massive model to the extension of 
a minimal set of mass-independent, almost homogeneously scaling distributions
by a Taylor expansion in the mass $m$. The sm-expansion is a generalization of 
this Taylor expansion, which yields this crucial simplification of the 
renormalization of massive models also for the case that one quantizes with 
the Wightman two-point function, which contains a $\log(-(m^2(x^2-ix^0 0))$-term.

We construct the general solution of the new system of axioms (i.e.~the usual 
axioms of causal perturbation theory completed by the sm-expansion),
and illustrate the method for a divergent diagram which contains a divergent 
subdiagram.
\end{abstract}

\maketitle

\section{Introduction}

In the inductive Epstein-Glaser construction of time-ordered products 
\cite{EpsteinG73,Stora82,BrunettiF00,BrunettiDF09} 
renormalization amounts to the extension of numerical distributions
$t^{(m)\,0}\in \sD'(\bR^k \less \{0\})$ to $t^{(m)}\in \sD'(\bR^k)$, where 
we assume translation invariance. By the upper index $m$ we denote the mass
of the underlying free theory. In the extension
$t^{(m)\,0}\to t^{(m)}$ one wants to maintain the property that $t^{(m)\,0}$ scales almost homogeneously 
under $(x,m)\to (\rho x,m/\rho)$ with a degree $D\in \bN$, i.e.
\be\label{eq:alm-hom-scal}
\Bigl(\sum_r x_r\del_r-m\del_m+D\Bigr)^N\,t^{(m)\,0}(x)=0
\ee
for a sufficiently large $N\in\bN$.
For an $m$-independent distribution $u^0\in \sD'(\bR^k \less \{0\})$, which scales almost homogeneously under
$x\to\rho x$ (i.e.~$u^0$ fulfils \eqref{eq:alm-hom-scal} without the $m\del_m$-term), quite a lot is known
about the extension to an $u\in \sD'(\bR^k)$ such that the almost homogeneous scaling is preserved
(see e.g.~proposition \ref{pr:extens-exist} and \cite{Hormander90,HollandsW02, DuetschF04,Hollands08}).
To profit from these knowledges, one wants to expand $t^{(m)\,0}(x)$ in terms of such distributions $u^0(x)$
(as done in \cite{HollandsW02,DuetschF04}). If $t^{(m)\,0}$ is smooth in $m\geq 0$, this expansion is simply 
the Taylor expansion in $m$ \cite{DuetschF04}:
\be\label{eq:Taylor-m}
t^{(m)\,0}(x)=\sum_{l=0}^L m^l\,u^0_l(x)+\mathfrak{r}^{(m)\,0}_{L+1}(x)\ .
\ee
Choosing $L$ sufficiently large, the remainder $\mathfrak{r}^{(m)\,0}_{L+1}\in \sD'(\bR^k \less \{0\})$
can easily be extended and one is left with the 
almost homogeneous extension of the $u^0_l$-distributions. This procedure maintains the scaling property
\eqref{eq:alm-hom-scal} and it also fulfils the renormalization condition 
$\sd(t^{(m)})=\sd(t^{(m)\,0})$, which is frequently used in causal perturbation theory. 
('$\sd$' means Steinmann's scaling degree \eqref{eq:sd-def},
 which is a measure for the UV-behavior of the distribution.)

If one quantizes the underlying free theory by using a {\it Hadamard function}
(which is smooth in $m\geq 0$), one can require smoothness in $m\geq 0$ as a 
renormalization condition for the time-ordered products and with that one can proceed as 
just described, see \cite{DuetschF04}.

However, mostly the {\it Wightman two-point function} $\Dl^+_m$ is used for the quantization.
In even dimensions $d$, $\Dl^+_m$ is not smooth in $m$ at $m=0$; for $d=4$ it is of the form
\be
\Dl^+_m(x)=\frac{-1}{4\pi^2(x^2 - ix^0 0)}
+ m^2\,f(m^2x^2)\log(-m^2(x^2 - ix^0 0)) + m^2\,F(m^2x^2),
\label{eq:Delta-plus} 
\end{equation}
with $f$ and $F$ being certain analytic functions. To
reduce renormalization to the extension of 
a minimal set of $m$-independent, almost homogeneously scaling distributions
also for time-ordered products based on quantization with
$\Dl^+_m$, we generalize \eqref{eq:Taylor-m} to
\be\label{eq:sm-expansion-0}
t^{(m)\,0}(x) = \sum_{l=0}^L m^l\,\sum_{p=0}^{P_l} \bigl( \log \tfrac{m}{M} \bigr)^p
\,u_{l,p}^0(x)+\mathfrak{r}^{(m)\,0}_{L+1}(x)\ ,\quad L,P_l\in\bN_0 ,
\ee
where $M>0$ is a fixed mass scale and $u_{l,p}^0,\,\mathfrak{r}^{(m)\,0}_{L+1}
\in \sD'(\bR^k \less \{0\})$. We call \eqref{eq:sm-expansion-0} the 'scaling and 
mass expansion'. This name refers to the following two possibilities to
interpret \eqref{eq:sm-expansion-0}: on the one hand it is an expansion in terms of 
$m$-independent, almost homogeneously scaling distributions $u_{l,p}(x)$ and on the other 
hand it is a ``Taylor expansion in the mass $m$ modulo $\log m$''. 

We require the sm-expansion for the $t^{(m)}$-distributions
as a new axiom for causal perturbation theory [sect.~3]. We will construct the general solution of 
the so modified system of axioms [sect.~4].

The sm-expansion
\eqref{eq:sm-expansion-0} is strongly related to the 'scaling expansion' of Hollands 
and Wald for time-ordered products on curved space-times \cite{HollandsW02}. A
main conceptual difference is that we require the structure \eqref{eq:sm-expansion-0}
directly as an axiom, whereas the 'scaling expansion' in \cite{HollandsW02} is a
non-trivial {\it consequence} of the system of axioms used there.

Working with a dimensionally regularized Feynman propagator as introduced in \cite{Keller13},
the sm-expansion \eqref{eq:sm-expansion-0} is of a different form:
$t^{(m)\,0}(x) = \sum_i(\tfrac{m}{M})^{z_i}\,u^0_i(x)$ plus a remainder, where
$z_i\in\bC$ and the $m$-independent distributions $u^0_i\in \sD'(\bR^k \less \{0\})$
scale even {\it homogeneously} with a degree $\kappa_i\in\bC\less\bZ$ [sect.~5].

We assume that the reader is familiar with the formalism for causal perturbation theory
introduced in \cite{DuetschF04}. 

\section{Axioms for causal perturbation theory}
 
\subsection{General axioms}

For simplicity we study a real scalar field $\vf$ on $d$-dimensional Minkowski space $\bM$, $d>2$.
On the space $\sF$ of observables (defined in \cite[formulas (2.1-2)]{DuetschF04}\footnote{Note
that the elements of $\sF$ are polynomials in $(\del^\beta)\vf$ and they are
formal power series in $\hbar$. The generalization to non-polynomial observables is 
given in \cite{BrunettiDF09}.})  
we introduce an $m$-dependent star product $\star_m : \sF \x \sF \to \sF$ \cite{DuetschF01a} by 
\begin{align}\label{eq:general-star}
F \star_m G :=& \sum_{n=0}^\infty \frac{\hbar^n}{n!} 
\int dx_1 \cdots dx_n\,dy_1 \cdots dy_n\,\notag\\ 
&\cdot\frac{\dl^n F}{\dl\vf(x_1)\cdots\dl\vf(x_n)}
\prod_{l=1}^n H_m(x_l - y_l) \,
\frac{\dl^n G}{\dl\vf(y_1)\cdots\dl\vf(y_n)} \,,
\end{align}
where 
\begin{itemize}
\item either $H_m=\Dl^+_m$ is the Wightman two-point function \eqref{eq:Delta-plus},
\item or $H_m=H^\mu_m$ is a {\it Hadamard function}, which depends on an additional mass 
parameter $\mu>0$. In even dimensions $d$, $H_m^\mu$ is related to $\Dl^+_m$ by
\begin{equation}
H_m^{\mu\,(d)}(x) := \Dl^{+\,(d)}_m(x) - m^{d-2}\,f^{(d)}(m^2x^2)\log(m^2/\mu^2)\ ,
\label{eq:Hadam-defn} 
\end{equation}
where $f^{(d)}$ is an analytic function which agrees for $d=4$ with the function $f$
in~\eqref{eq:Delta-plus} (see \cite[Appendix A]{DuetschF04}).
Thus, the $\log(-m^2(x^2 - ix^0 0))$ factor in~\eqref{eq:Delta-plus}
is replaced by $\log(-\mu^2(x^2 - ix^0 0))$, due to that $H^\mu_m$ is 
smooth in $m\geq 0$. 
\end{itemize}
In both cases $H_m$ is a Lorentz invariant 
solution of the Klein-Gordon equation; the antisymmetric part of $H_m$ is fixed by 
$H_m(x)-H_m(-x)=i\,\Dl_m(x)$ (where $\Dl_m$ is the commutator function).

Let $\sP$ be the space of polynomials in $\del^\beta\vf$, for $\beta \in \bN_0^d$.
Following \cite{DuetschF04}, a \textbf{time-ordered product} $T^{(m)}\equiv T=(T_n)_{n\in\bN}$ 
($m$ denotes the mass of the underlying star product) is a sequence of maps 
$T_n\: \sP^{\ox n} \to \sD'(\bM^n,\sF)\ $,\footnote{Note that both the 
arguments and the values of $T_n$
are {\it off-shell fields}, i.e.~not restricted by any field equation.} 
which are \textbf{linear}; and satisfy
\begin{itemize}
\item[(a)] \textbf{Initial value:}
$T_1(A(x)) = A(x)$ for any $A \in \sP$;
\item[(b)] \textbf{Permutation symmetry:}
$T_n(A_{\pi(1)}(x_{\pi(1)}),...,A_{\pi(n)}(x_{\pi(n)}))$ 
\begin{align*}
\qquad\qquad\qquad\qquad= T(A_1(x_1),...,A_n(x_n))\quad\forall\pi \in S_n\ ;
\end{align*}
\item[(c)] \textbf{Causality:} $\quad T_n(A_1(x_1),\dots,A_n(x_n))$
\begin{align}
\qquad\qquad\qquad&= T_k(A_1(x_1),\dots,A_k(x_k)) \star_{m}
T_{n-k}(A_{k+1}(x_{k+1}),\dots, A_n(x_n))
\label{eq:Tprods-causal} 
\\
&\text{whenever}\quad
\{x_1,\dots,x_k\} \cap \bigl( \{x_{k+1},\dots,x_n\} + \ovl V_- \bigr)
= \emptyset.
\nonumber
\end{align}
\end{itemize}
These are the basic axioms. In the inductive step $\{T_1,...,T_{n-1}\}\to T_n$
of the construction of the sequence $T$, these axioms determine 
\be\label{eq:T-0}
T^0_n(A_1(x_1),...):= T_n(A_1(x_1),...)\vert_{\sD(\bM^n \less \Dl_n)}
\ee
\emph{uniquely}, where
$\Dl_n := \set{(x_1,\dots,x_n) \in \bM^n : x_1 = \dots= x_n}$
is the thin diagonal. 

The further axioms (called 'renormalization conditions') restrict only the
\textit{extension} to $\sD'(\bM^n,\sF)$.  
\begin{itemize}
\item[(d)]
\textbf{Field independence}:
\begin{equation}
\fd{}{x} T_n(A_1(x_1),...,A_n(x_n)) = \sum_{l=1}^n
T_n\biggl( A_1(x_1),..., \fd{A_l(x_l)}{x},...,A_n(x_n) \biggr).
\label{eq:field-indep} 
\end{equation}
Using this property in a (finite) Taylor expansion of 
$T_n( A_1(x_1),\dots)$ w.r.t. $\vf=0$, one obtains the \textbf{causal Wick 
expansion}: for \textit{monomials} $A_1,\ldots,A_n\in\sP$ it holds
\begin{align}
T_n\bigl( A_1(x_1),\dots, A_n(x_n) \bigr)
= &\sum_{\unl A_l \subseteq A_l}\! \om_0\bigl( 
T_n\bigl( \unl A_1(x_1),\dots, \unl A_n(x_n) \bigr) \bigr)\,\notag\\
& \cdot\ovl A_1(x_1) \cdots \ovl A_n(x_n),
\label{eq:causal-Wick-expan} 
\end{align} 
where $\om_0:F\mapsto\om_0(F):=F\vert_{\vf=0}$ denotes the vacuum state. In addition,
each \textit{submonomial} $\unl A$ of a given monomial $A$ and its 
\textit{complementary submonomial} $\ovl A$ are defined by
\be\label{subpolynomials}
\unl A
:= \frac{\del^k A}{\del(\del^{\beta_1}\vf)\cdots\del(\del^{\beta_k}\vf)}
\neq 0\ ,  \quad
\ovl A := C_{\beta_1\dots \beta_k} \,\del^{\beta_1}\vf \cdots \del^{\beta_k}\vf
\ee
(no sum over $\beta_1,...,\beta_k$),
where each $C_{\beta_1\dots \beta_k}$ is a certain combinatorial factor and the range of the sum
$\sum_{\unl A \subseteq A}$ are all allowable $k$ and $\beta_1,...,\beta_k$.
(For $k=0$ we have $\unl A=A$ and $\ovl A=1$.)

\item[(e)]
\textbf{Translation invariance}: the $\bC$-valued distributions
\begin{equation}
t^{(m)}(A_1,\dots,A_n)(x_1 - x_n,\dots, x_{n-1} - x_n)
:= \om_0\bigl( T^{(m)}_n\bigl(
A_1(x_1),\ldots, A_n(x_n) \bigr) \bigr)
\label{eq:t} 
\end{equation}
depend {\it only} on the relative coordinates.

\item[(f)]
\textbf{Action Ward Identity (AWI)}:
\be\label{AWI}
\del_{x_k^\mu} T_n\bigl( A_1(x_1),..., A_k(x_k),... \bigr)
=T_n\bigl( A_1(x_1),..., \del_\mu A_k(x_k),... \bigr)\ .
\ee
\end{itemize}

The axioms (d) and (e) simplify the extension $T_n^0(A_1,...)\to T_n(A_1,...)$
to the problem of extending the $\bC$-valued distributions
$t^0(A_1,\dots)(x_1 - x_n,\dots)$ $:= \om_0\bigl( T^0_n(
A_1(x_1),\ldots ) \bigr)\in \sD'(\bR^{d(n-1)} \less \{0\})$ to 
$t(A_1,\dots)(x_1 - x_n,\dots)\in $ $\sD'(\bR^{d(n-1)}),\,\,\forall A_1,...,A_n\in\sP$.

The AWI can be fulfilled by using that 
there exists a subspace $\sP_\bal\subset\sP$ (called 'balanced fields') such that
every $A \in \sP$ can \textit{uniquely} be written as a finite sum
$$
A = \sum_k \del^{\beta_k} B_k\quad\text{where}\quad B_k \in \sP_\bal\ ,\quad\beta_k\in\bN_0^d
$$
(see \cite[Sect.~3.2]{DuetschF04} for the definition of $\sP_\bal$). 
Since $t^0$ fulfills the AWI by induction, one can proceed as follows:
one constructs the extension $t(B_1,...,B_n)$ first only for all balanced fields
$B_1,...,B_n\in\sP_\bal$. Then, using linearity of $T_n$ and writing arbitrary
$A_1,...,A_n\in\sP$ as $A_i = \sum_{k_i} \del^{\beta_{ik_i}} B_{ik_i}$ (where $B_{ik_i} \in \sP_\bal$), 
the definition
\begin{align}\label{eq:AWI-extension}
t(A_1,\ldots,A_n)(x_1-x_n,\ldots)&:=\sum_{k_1,...,k_n}\notag\\
\del^{\beta_{1k_1}}\cdots\del^{\beta_{nk_n}}&
t(B_{1k_1},\ldots,B_{nk_n})(x_1-x_n,\ldots)\ .
\end{align}
yields indeed an extension of $t^0(A_1,\ldots,A_n)$ which satisfies the AWI.
 
\begin{itemize}
\item[(g)]
\textbf{Scaling}: The mass dimension of a field monomial is defined by
\be\label{eq:mass-dimension}
\mathrm{dim}\,\prod_{j=1}^J \partial^{\beta_j}\vf:=J\,\tfrac{d-2}2+\sum_{j=1}^J|\beta_j|\ .
\ee
Let $\sP_\homog$ be the set of ``homogeneous'' polynomials, i.e.~an $A\in\sP_\homog$
is a linear combination of monomials which have the {\it same mass dimension}.

The scaling axioms requires that for $A_1,\dots,A_n \in \sP_\homog$ the numerical 
distributions \eqref{eq:t} scale almost homogeneously under 
$(x,m)\to (\rho x, m/\rho)\ $,\footnote{When quantizing with a Hadamard function 
$H^\mu_m$, the mass parameter $\mu$ is not scaled.} that is
\be\label{eq:scalingprop}
0=(\rho\,\del_\rho)^N\Bigr(\rho^D t^{(m/\rho)} \bigl( 
A_1,\dots, A_n \bigr)(\rho x) \Bigr) 
\ee
for a sufficiently large $N\in\bN$, where the degree $D$ is given by 
$D := \sum_{k=1}^n\dim A_k \in\bN\ .$ 
That $D$ is a natural 
number follows from the observation that $t^{(m)}(A_1,\dots, A_n)$
is non-vanishing only if the number of basic fields
$\del^\beta \vf$ in $\{A_1,\dots,A_n\}$ is even.

By the 'power' of the almost 
homogeneous scaling we mean $N-1$ for the minimal $N\in\bN$ fulfilling
\eqref{eq:scalingprop} (or equivalently \eqref{eq:alm-hom-scal}).

\item[(h)] The axioms \textbf{Lorentz covariance, unitarity, off-shell field equation} and 
\textbf{symmetries} are not relevant for our purposes, hence, we do not explain them here.
\end{itemize}

\subsection{Axioms for quantization with a Hadamard function}

In this subsection we assume that quantization is done by a 
Hadamard function $H^\mu_m$. Then the star product $\star_{m,\mu}$ and, via the causality axiom,
the time-ordered product $T^{(m,\mu)}$ depend on $\mu$.
We complete the system of axioms as follows \cite{DuetschF04}:
\begin{itemize}
\item[(i)]
\textbf{Smoothness in the mass $m \geq 0$}: Since $H^\mu_m$ is smooth in $m\geq 0$,
we may require that the functions
\be\label{eq:smooth-m}
0 \leq m \longmapsto \langle t^{(m,\mu)}(A_1,...,A_n)\,,\,g\rangle
\word{be smooth} \forall A_1,...,A_n\in\sP
\ee
and $\forall g\in\sD(\bR^{d(n-1)})$.
\item[(j)]
\textbf{$\mu$-covariance}: Let
\begin{align}\label{eq:Gamma}
&\Ga := \int dx\,dy\, m^{d-2}\,f^{(d)}(m^2(x - y)^2)\,
\frac{\dl^2}{\dl\vf(x)\,\dl\vf(y)} \word{and}\notag\\
&r^\Ga := 1 + \sum_{k=1}^\infty \frac{1}{k!}\, ((\log r) \Ga)^k,
\end{align}
where $r>0$ and the function~$f^{(d)}$ is the one that appears in the definition
\eqref{eq:Hadam-defn} of the Hadamard function.
With that the operator $(\tfrac{\mu_2}{\mu_1})^\Ga$
intertwines the different star products for $\mu_1$ and $\mu_2$:
\be\label{eq:star-equivalent}
F \star_{m,\mu_2} G
= (\tfrac{\mu_2}{\mu_1})^\Ga \Bigl( ((\tfrac{\mu_2}{\mu_1})^{-\Ga} F)
\star_{m,\mu_1} ((\tfrac{\mu_2}{\mu_1})^{-\Ga} G) \Bigr)\ .
\ee
We require the same relation for the time-ordered products:
\begin{align}\label{eq:mu-cov}
&T_n^{(m,\mu_2)}\bigl(A_1(x_1),...,A_n(x_n)\bigr)=   \notag\\
&( \tfrac{\mu_2}{\mu_1} )^\Ga
\Bigl( T_n^{(m,\mu_1)}\bigl((\tfrac{\mu_2}{\mu_1})^{-\Ga}A_1(x_1),...,
(\tfrac{\mu_2}{\mu_1})^{-\Ga}A_n(x_n)\bigr)\Bigr).
\end{align}
\end{itemize}

\subsection{Modification of the axioms such that the Wightman two-point
function is admitted}

Smoothness in $m \geq 0$, axiom~(i),  excludes the
Wightman two-point function $\Dl^+_m$ in even dimensions $d$. However, a 
time-ordered product $(T^{(m)}_n)_{n\in\bN}$ based on quantization with
$\Dl^+_m$ can be axiomatically defined by using that the operator 
$(\tfrac{\mu}{m})^\Ga$ intertwines the star products $\star_m$ (based on $\Dl^+_m$)
and $\star_{m,\mu}$ (based on $H^\mu_m$). (This statement is 
obtained by inserting $H^m_m=\Dl^+_m$ into \eqref{eq:star-equivalent}.)
Due to that one may
replace axiom (i) by the requirement that the transformed time-ordered product
\be\label{eq:transformedT}
( \tfrac{\mu}{m} )^{\Ga}
\Bigl( T_n^{(m)}\bigl((\tfrac{\mu}{m})^{-\Ga}A_1(x_1),...,
(\tfrac{\mu}{m})^{-\Ga}A_n(x_n)\bigr)\Bigr)
\ee
be  smooth in $m \geq 0$, as done in \cite{DuetschF04,Keller13}. 
(That is, the vacuum expectation values
$t^{(m,\mu)}(A_1,...,A_n)$ $:=\om_0\bigl($\eqref{eq:transformedT}$\bigr)$
fulfil \eqref{eq:smooth-m}.)
In addition, the $\mu$-covariance, axiom (j), is unnecessary, it has to be omitted; 
all other axioms remain unchanged.

Since smoothness in $m\geq 0$ is very helpful for the construction of the 
time-ordered products (by means of the Taylor expansion \eqref{eq:Taylor-m}), the 
obvious way to construct a solution of the so modified
system of axioms is, to construct first the time-ordered product 
$(T^{(m,\mu)}_{n})_{n\in\bN}$ (which is based on $H^\mu_m$), and then
$(T^{(m)}_n)_{n\in\bN}$ is obtained by the inverse transformation of 
\eqref{eq:transformedT}. 

Following essentially \cite{DuetschF04}, we explain why this construction fulfils the axiom
(g) (scaling). First, for $t_H^{(m,\mu)}:=\om_0(T^{(m,\mu)})$ this is obtained as follows:
using causality and the inductive assumption one shows that
$t_H^{(m,\mu)\,0}\in \sD'(\bR^k \less \{0\})$ fulfils 
\eqref{eq:alm-hom-scal}, this is analogous to our procedure in sect.~4.1.
It follows that the pertinent distributions $u_l^0\in \sD'(\bR^k \less \{0\})$
in the Taylor expansion \eqref{eq:Taylor-m} scale almost homogeneously with degree $D-l$.
The extension $u_l^0\to u_l\in\sD'(\bR^k)$ is done such that this property is maintained.
Therefore, inserting $u_l$ into \eqref{eq:Taylor-m}, we obtain that the resulting 
$t_H^{(m,\mu)}\in \sD'(\bR^k)$ fulfils \eqref{eq:alm-hom-scal} (or equivalently 
\eqref{eq:scalingprop}).\footnote{The remainders in the Taylor expansion \eqref{eq:Taylor-m}
are teated in the same way as in our construction in sect.~4.2, hence we neglect them here.}

The second step is to verify that the axiom (g) is preserved in the 
inverse transformation of \eqref{eq:transformedT}:\footnote{In \cite{DuetschF04}
this verification is done in terms of a scaling transformation $\sigma_\rho$, which is an
algebra isomorphism from $(\sF,\star_{(\rho^{-1}m,\rho^{-1}\mu)})$ to $(\sF,\star_{(m,\mu)})$.
To minimize the mathematical tools, we do not introduce $\sigma_\rho$ in this paper.}
we use that $t^{(m)}:=\om_0(T^{(m)})$ can be written as $t^{(m)}=t_H^{(m,m)}$ (due to 
$H^m_m=\Dl^+_m$). With that, the assertion
$$
0=(\rho\,\del_\rho)^N\Bigr(\rho^D t^{(m/\rho)} (A_1,\dots)(\rho x) \Bigr)=
(\rho\,\del_\rho)^N\Bigr(\rho^D 
t_H^{(m/\rho,m/\rho)} ( A_1,\dots)(\rho x) \Bigr)
$$
can equivalently be written as
\be\label{eq:scal-t_H}
0=\Bigl(x\,\del_x-m\del_m+D-\mu\,\del_\mu\Bigr)^N\,t_H^{(m,\mu)}(A_1,\ldots)(x)\vert_{\mu=m}\ ,
\ee
where $x\,\del_x:=\sum_r x_r\,\del_{x_r}$. From \eqref{eq:mu-cov} we see that
\be
(\mu\,\del_\mu)^K t_H^{(m,\mu)}(A_1,\ldots)=(\mu\,\del_\mu)^K\om_0\Bigl(
( \tfrac{\mu}{\mu_0} )^\Ga\circ
T^{(m,\mu_0)}\bigl((\tfrac{\mu}{\mu_0})^{-\Ga}A_1,...\bigr)\Bigr)=0
\ee
for $K\in\bN$ sufficiently large; namely, since our functionals 
$F\in\sF$ are {\it polynomials} in $(\del^\beta)\vf$, an expression
$r^\Gamma F$ \eqref{eq:Gamma} is a polynomial in $\log r$.
Now writing the r.h.s.~of \eqref{eq:scal-t_H} as
$$
\sum_{K=0}^{N} \binom{N}{K}
(x\,\del_x-m\,\del_m + D)^{N-K} (-\mu\,\del_\mu)^{K}\, 
t_H^{(m,\mu)}(A_1,\ldots)(x)\vert_{\mu=m}\ ,
$$
we see that this expression vanishes indeed for $N\in\bN$ sufficiently large.

\begin{ex}
We illustrate for the setting sun diagram in $d=4$ dimensions how
$t^{(m)}(\vf^3,\vf^3)$ (based on $\Dl^+_m$) can be obtained from
$T^{(m,\mu)}$-terms in practice. From  \eqref{eq:Hadam-defn} we know that the
Feynman(-like) propagators fulfil
$$
\Dl^F_m(x)=H^{F,\mu}_m(x)+d^\mu_m(x)\quad\text{with}\quad d^\mu_m\in C^\infty\ ,
$$
where $H^{F,\mu}_m(x):=\theta(x^0)H^\mu_m(x)+\theta(-x^0)H^\mu_m(-x)$. Inserting this into 
\newline
$t^{(m)}(\vf^3,\vf^3)(x)=6\,\hbar^3\,(\Dl^F_m(x))^3$ we obtain
\begin{align*}
t^{(m)}(\vf^3,\vf^3)(x)=&
t^{(m,\mu)}(\vf^3,\vf^3)(x)+9\,\hbar\,t^{(m,\mu)}(\vf^2,\vf^2)(x)\,\,d^\mu_m(x)\notag\\
&+18\,\hbar^3\,H^{F,\mu}_m(x)\,(d^\mu_m(x))^2+6\,\hbar^3\,(d^\mu_m(x))^3\ .
\end{align*}
Since $d^\mu_m$ is smooth, all appearing pointwise products exist.
\end{ex}

However, in view of a {\it direct} construction of $(T^{(m)}_n)_{n\in\bN}$, we are 
searching a direct axiomatic definition of these objects.
We want to keep almost homogeneous scaling (with degree $D$) of the
distributions $t\equiv t^{(m)}(A_1,\dots,A_n)$, $A_1,\dots A_n\in\sP_\mathrm{hom}$, 
see \eqref{eq:scalingprop}. This axiom admits the addition of a term
\be\label{eq:nonunique-scaling}
t(x_1 - x_n,\dots,x_{n-1} - x_n) + \sum_{|\beta|+l = D-d(n-1)} \!
m^l\,C^{(m)}_{l,\beta} \,\del^\beta\dl(x_1 - x_n,\dots,x_{n-1} - x_n),
\ee
where $l \in \bZ$ (since $D \in \bN$) and the numbers $C^{(m)}_{l,\beta}\in\bC$ are,
as functions of $m$, 
polynomials in $\log(m/M)$, where $M>0$ is some renormalization mass
scale. But to fulfil the usual requirement $\sd(t) = \sd(t^0)$ on
extensions $t$ of~$t^0$, we need a substitute for smoothness in
$m\geq 0$, which excludes negative values of~$l$. Such a candidate is:
\begin{itemize}
\item[(i$'$)]
\textbf{Continuity in the mass $m \geq 0$}:
We require that the functions
\be\label{eq:continuous-m}
0 \leq m \longmapsto \langle t^{(m)}(A_1,...,A_n)\,,\,g\rangle
\word{be continuous} \forall A_1,...,A_n\in\sP
\ee
and $\forall g\in\sD(\bR^{d(n-1)})$.
\end{itemize}

With that, the Wightman two-point function $\Dl^+_m$ is admitted also
in even dimensions~$d$. (Recall that $\Dl^+_m$ is actually $C^1$ in
$m \geq 0$.)

\begin{rem}[central solution and mass-shell renormalization] 
If all fields are massive (i.e., $m > 0$), any admissible extension
$t^{(m)} \in \sD'(\bR^k)$ of a given $t^{(m)\,0} \in \sD'(\bR^k \less \{0\})$
has the property that its Fourier transformed\footnote{Fourier transformation
is meant w.r.t.~the relative coordinates $x\equiv(x_1-x_n,...,x_{n-1}-x_n)$.}
distribution $\hat t^{(m)}(p)$ is analytic in a neighbourhood of $p=0$ 
(see \cite{EpsteinG73}). Therefore, the so-called ``central solution
$t_c^{(m)}$'' of the extension problem exists, which is 
defined by
$$
\del^\beta \hat t_c^{(m)}(0)=0\quad\forall |\beta|\leq\om\ ,\quad
\om :=\sd(t^{(m)\,0})-k\ .
$$ 
It can be obtained from any extension $t^{(m)}$ with 
$\sd(t^{(m)})=\sd(t^{(m)\,0})$, by Taylor subtraction:
$$
\hat t_c^{(m)}(p)=\hat t^{(m)}(p)-\sum_{|\beta|\leq \om}\frac{p^\beta}{\beta!}
\,\del^\beta\hat t^{(m)}(0)\ ,
$$
which corresponds to ``BPHZ-subtraction at $p=0$". We conclude: if
there exists an extension $t^{(m)}$ which fulfills the scaling axiom
\eqref{eq:scalingprop} with degree $D=\om+k\in\bN$ and power $(N-1)$, i.e.
$$
(\rho\del_\rho)^N\,\Bigl(\rho^{\om}\,\hat t^{(m/\rho)}(p/\rho)\Bigr)=0\ ,
$$
then, this holds also for $t_c^{(m)}$ with the same degree 
and the same power. But, it is well known that the limit 
$\lim_{m\downto 0} t_c^{(m)}$ diverges in general,%
\footnote{This holds e.g.~for the fish diagram in $d = 4$
dimensions.}
i.e.~the central solution is in conflict with continuity in $m \geq 0$
and, hence, also with the sm-expansion 
axiom (which is treated in the following sections).

To discuss mass-shell renormalization we study a $\vf^4$-interaction
in $d=4$ dimensions (or $\vf^3$ in $d=6$). Let 
$$
\Sigma^m_n(p^2):=\hat t^{(m)}(\vf^3,\vf^4,...,\vf^4,\vf^3)(p,0,...,0)
$$
(or the same for $\hat t(\vf^5,\vf^6,...,\vf^6,\vf^5)$ in the ($d=6$)-case)
be the self-energy contribution to $n$-th order; it has $\om=2$ to all orders. 
The inner momenta $p_j$ ($j=2,...,n-1$) are set to $p_j=0$, due to integrating out the 
inner vertices $x_j$ with $g(x_j)\equiv 1$ (``partial adiabatic limit'', see
e.g.~\cite{Duetsch97}). We use the notation
$\Sigma^{m\,\prime}_n(p^2):=\tfrac{\del}{\del p^2}\Sigma^m_n(p^2)$. In addition,
let $m_0$ be the physical mass. The  mass-shell renormalization
 $\Sigma^m_{n,m_0}(p^2)$ is uniquely defined by 
$$
\Sigma^m_{n,m_0}(m_0^2)=0\quad\text{and}\quad \Sigma^{m\,\prime}_{n,m_0}(m_0^2)=0\ ,
\quad\forall n\geq 2\ ,
$$
and is obtained by Taylor subtraction (``BPHZ-subtraction at $p^2=m_0^2$"):
$$
\Sigma^m_{n,m_0}(p^2)=\Sigma^m_n(p^2)-\Sigma^m_n(m_0^2)
-(p^2-m_0^2)\,\Sigma^{m\,\prime}_n(m_0^2)\ .
$$
If $\Sigma^m_n(p^2)$ scales almost homogeneously with power $N_n-1$,
i.e. 
$$
(\rho\del_\rho)^{N_n}\,\bigl(\rho^2\,\Sigma^{m/\rho}_n(p/\rho)\bigr)=0\ ,
\,\,\,\text{then generally}\,\,\, 
(\rho\del_\rho)^N\bigl(\rho^2\,\Sigma^{m/\rho}_{n,m_0}
(p/\rho)\bigr)\not=0
$$ 
$\forall N\in\bN$, because the subtraction point $m_0^2$ is not 
scaled. However, usually one sets $m:=m_0$ (``mass renormalization'') and, if
$m_0$ is also scaled, we obtain
$$
(\rho\del_\rho)^{N_n}\,\Bigl(\rho^2\,\Sigma^{m_0/\rho}_{n,m_0/\rho}
(p/\rho)\bigr)=0\ .
$$
But, the limit $\lim_{m_0\downto 0}\Sigma^{m_0}_{n,m_0}$ diverges in general,
because the central solution $\Sigma^m_{n,0}(p^2)=\hat t^{(m)}_c(p,0,...,0)$ 
generally does not exist for $m=0$.
\end{rem}

\section{The scaling and mass expansion}

The difficult question is: how to construct a solution of the
just proposed system of axioms (a)-(h) and (i$'$)? 
We solve the problem in an indirect way, by 
replacing the almost homogeneous scaling, axiom (g), and the continuity in $m \geq 0$, axiom~(i$'$),
by the following new axiom:

\begin{itemize}
\item[(k)]
\textbf{Scaling and mass expansion}:
For all field {\it monomials} $A_1,\dots,A_n\in \sP$, the vacuum expectation values
$t^{(m)}(A_1,\dots,A_n)(x_1 - x_n,\dots, x_{n-1} - x_n)$ \eqref{eq:t}
fulfil the sm-expansion with degree  
$D := \sum_{k=1}^n \dim A_k$, where the following definition is used:

\begin{defn} 
\label{df:sm-expansion}
A distribution $f^{(m)}\in\sD'(\bR^k)$ or $f^{(m)}\in
\sD'(\bR^k\setminus\{0\})$, depending on $m\geq 0$,
 fulfils the sm-expansion 
with degree $D$, if for all $l,L\in\bN_0$ there exist distributions
$u^{(m)}_l ,\,\,\mathfrak{r}^{(m)}_{L+1}\in\sD'(\bR^k[\setminus\{0\}])$ such that
\be\label{eq:sm-expansion}
f^{(m)}(x) = \sum_{l=0}^L m^l\,u^{(m)}_l(x)+\mathfrak{r}^{(m)}_{L+1}(x)\quad
\forall L\in\bN_0\ ,
\ee
and
\begin{enumerate}
\item[(A)]
$u_0 \equiv u_0^{(m)}$ is independent of~$m$ and $u_0=f^{(0)}$;
\item[(B)] For $l\geq 1$ the $m$-dependence of
$u^{(m)}_l(x)$ is a polynomial in $ \log \tfrac{m}{M}$, 
where $M>0$ is a fixed mass scale. Explicitly, there exist $m$-independent
distributions $u_{l,p} \in \sD'(\bR^k [\less \{0\}])$ such that
\begin{equation}
u^{(m)}_l(x) = \sum_{p=0}^{P_l} \bigl( \log \tfrac{m}{M} \bigr)^p
\,u_{l,p}(x), \quad P_l < \infty\ .
\label{eq:log(m)-expansion} 
\end{equation}
(Of course, the distributions $u_{l,p}$ depend on $M$.)
\item[(C)]
$u^{(m)}_l(x)$ scales almost homogeneously in~$x$ with degree $D - l$ and,
hence, this holds also for all $u_{l,p}$ \eqref{eq:log(m)-expansion}; 
\item[(D)]
$\mathfrak{r}^{(m)}_{L+1}(x)$ is almost homogeneous with degree $D$ under
the scaling $(x,m)\mapsto (\rho x, m/\rho)$;
\item[(E)] $\mathfrak{r}^{(m)}_{L+1}$ is smooth in $m$ for $m>0$ and
$$
\lim_{m\downarrow 0}\,(\tfrac{m}{M})^{-(L+1)+\eps}\,\,\mathfrak{r}^{(m)}_{L+1}=0\quad\quad\forall\eps >0\ .
$$
\end{enumerate}
(All properties are meant in the weak sense, e.g.~(E) holds for
$\langle \mathfrak{r}^{(m)}_{L+1},h\rangle\quad$ 
$\forall h\in \sD(\bR^k[\less \{0\}])$.)
\end{defn}
\end{itemize}

As explained after \eqref{eq:scalingprop}, the degree $D=\sum_k\dim A_k$ is a natural number.

One easily verifies that, in $d=4$ dimensions, the Wightman two-point function
$\Dl^+_m$ \eqref{eq:Delta-plus} fulfils the sm-expansion with degree $D=2$.
For arbitrary $d\geq 3$, $\Dl^{+\,(d)}_m$ fulfils the sm-expansion with 
degree $D=d-2$. (If $d$ is odd, $\Dl^{+\,(d)}_m$ is smooth in $m\geq 0$, hence 
the sm-expansion is simply the Taylor expansion.)
Taking additionally $\dim\vf=\tfrac{d-2}2$ into account, we find  that
$$
t^{(m)}(\vf(x_1),\vf(x_2)) = \hbar\,\Dl^F_m(y)
= \hbar\,\bigl(\th(y^0)\Dl^+_m(y)+\th(-y^0)\Dl^+_m(-y)\bigr) 
$$
(where $y \equiv x_1 - x_2$) fulfils the new axiom (k).

The following lemma gives basic properties of distributions fulfilling
the sm-expansion.

\begin{lem} 
\label{lm:sm-expans-properties}
We assume that $f^{(m)}\in\sD'(\bR^k[\setminus\{0\}]),\,
f^{(m)}_1\in\sD'(\bR^{pd}[\setminus\{0\}])$ and $f^{(m)}_2\in\sD'(\bR^{qd}[\setminus\{0\}])$ satisfy the 
definition~\ref{df:sm-expansion} with degree $D$, $D_1$ or $D_2$, 
respectively. Then the following statements hold true:
\begin{enumerate}
\item[(1)] $f^{(m)}$ is smooth in $m$ for $m>0$ and
$\lim_{m\downto 0}\,f^{(m)} = u_0=f^{(0)}\ $.
\item[(2)]
$f^{(m)}(x)$ is almost homogeneous with degree $D$ under
the scaling $(x,m)\mapsto (\rho x, m/\rho)$.
\item[(3)]
$\partial^\beta_x\, f^{(m)}(x)$ (where $\beta$ is a multi-index)
fulfils the sm-expansion with degree $D+|\beta|$. 
\item[(4)]
We assume that the product of
distributions $f^{(m)}_1(x) f^{(m)}_2(y)$, which may be a 
(partly) pointwise product\footnote{More precisely: let $(x_1,...,x_p)$ and
$(y_1,...,y_q)$ (where $x_i,y_j\in\bR^d$) be the linearly independent components of $x\in\bR^{pd}$
and $y\in\bR^{qd}$, respectively. Then, the set $\{x_1,...,x_p,y_1,...,y_q\}$ may be linearly dependent.}, 
exists. Then, $f^{(m)}_1(x)\, f^{(m)}_2(y)$
fulfils also the sm-expansion with degree
$D = D_1 + D_2$.
\item[(5)] The sm-expansion is {\it unique}, i.e.~if
we know that a given $f^{(m)}$ has such an expansion, then the ``coefficients''
$u^{(m)}_l$ (and, hence, also the ``remainders'' $\mathfrak{r}^{(m)}_{L+1}$) are uniquely determined.
\item[(6)] The scaling degree of the remainder is bounded by $\,\,\sd(\mathfrak{r}^{(m)}_{L+1})\leq D-(L+1)$.
\end{enumerate}
\end{lem}

\begin{proof}
{\it Part (1)} follows immediately from \eqref{eq:sm-expansion}
and properties (A),(B) and (E). 

{\it Part (2)}: we have 
to show that $m^l\,u_l^{(m)}(x)$ has the asserted 
scaling property. This can be done as follows:
\begin{align*}
& (x\,\del_x + D - m\,\del_m)^{N}\, m^l\,u_l^{(m)}(x)
= m^l\,(x\,\del_x + (D - l) - m\,\del_m)^{N}\, u_l^{(m)}(x)
\\
&\quad = m^l\ \sum_{k=0}^{N} \binom{N}{k}
(x\,\del_x + D - l)^k (-m\,\del_m)^{N-k}\, u_l^{(m)}(x)\ ,
\end{align*}
where $x\,\del_x:=\sum_{i=1}^k x_i\del_{x_i}$.
Now, choosing $N$ sufficiently large, at least one of the operators 
$(x\,\del_x + D - l)^k$ or $(-m\,\del_m)^{N-k}$ yields zero when applied 
to $u_l^{(m)}(x)$, due to properties (C) and (B), respectively.

{\it Part (3)}: we show that $\del^\beta_x u_l^{(m)}(x)$ and  
 $\del^\beta_x \mathfrak{r}^{(m)}_{L+1}(x)$ satisfy the properties (A)-(E) with degree
$D+|\beta|$. To verify (D) let $N\in\bN$ be such that 
$(x\,\del_x + D - m\,\del_m)^N\,\mathfrak{r}^{(m)}_{L+1}(x)=0$. It follows that
$$
0=\del^\beta_x\,(x\,\del_x + D - m\,\del_m)^N\,\mathfrak{r}^{(m)}_{L+1}(x)
=(x\,\del_x + D+|\beta| - m\,\del_m)^N\,\del^\beta_x \,\mathfrak{r}^{(m)}_{L+1}(x).
$$
(C) can be shown analogously. To verify (A), (B) and (E) we use that these properties 
hold for $\langle g^{(m)},h\rangle$, where $g^{(m)}=u^{(m)}_l$ or
$g^{(m)}= \mathfrak{r}^{(m)}_{L+1}$, for all $h\in \sD(\bR^k[\less \{0\}])$. 
Hence, they hold for $(-1)^{|\beta|}\,\langle g^{(m)},\del^\beta h\rangle
=\langle \del^\beta g^{(m)},h\rangle\,\,\forall h$.

{\it Part (4)}: by a straightforward calculation we obtain
$$
f^{(m)}_1(x)\, f^{(m)}_2(y)
= \sum_{l=0}^L m^l\, u^{(m)}_l(x,y)+\mathfrak{r}^{(m)}_{L+1}(x,y)\ ,
$$
where
\begin{align*}
u^{(m)}_l(x,y)&:= \sum_{k=0}^l u^{(m)}_{1,k}(x)\,u^{(m)}_{2,l-k}(y),\quad (0\leq l\leq L)\\
\mathfrak{r}^{(m)}_{L+1}(x,y)&:=\mathfrak{r}^{(m)}_{1,L+1}(x)\,\mathfrak{r}^{(m)}_{2,L+1}(y)+
\mathfrak{r}^{(m)}_{1,L+1}(x)\,\sum_{l=0}^L m^l\,u_{2,l}^{(m)}(y)\\
&+\bigl(\sum_{l=0}^L m^l\,u_{1,l}^{(m)}(x)\bigr)\,\mathfrak{r}^{(m)}_{2,L+1}(y)
+\sum_{l=L+1}^{2L}m^l \sum_{k=l-L}^L  u^{(m)}_{1,k}(x)\,u^{(m)}_{2,l-k}(y)\ .
\end{align*}
With that, it is an easy task to verify that $u^{(m)}_l(x,y)$ and
$\mathfrak{r}^{(m)}_{L+1}(x,y)$ satisfy the properties (A)-(E) with degree $D=D_1+D_2$,
by using that $u^{(m)}_{j,l}$ and
$\mathfrak{r}^{(m)}_{j,L+1}$ fulfil these properties with degree $D_j$ (where $j=1,2$).

{\it Part (5)}: the determination of $u_0$ is given in part (1). For $l\geq 1$
we assume that $u^{(m)}_k$ is known for $k<l$ and we determine the coefficients
$u_{l,p}$ of $u^{(m)}_l$ \eqref{eq:log(m)-expansion} as follows: for $\bN\ni P>P_l$ 
the limit
\be\label{eq:m-limit}
\lim_{m\downto 0} \,\bigl( \log \tfrac{m}{M} \bigr)^{-P}\,m^{-l}\,\Bigl(
f^{(m)}(x)-\sum_{k=0}^{l-1} m^k\,u_k^{(m)}(x)\Bigr)
\ee
gives zero, for $P=P_l$ it gives $u_{l,P_l}$ and for $P<P_l$ it diverges. 
Since $P_l$ is unknown, we start with a $P$ which is sufficiently high
that the limit exists, if it vanishes we lower $P$ by $1$ etc.. Having determined
$P_l$ and $u_{l,P_l}$ in this way, we compute 
\begin{align*}
&\lim_{m\downto 0} \,\bigl( \log \tfrac{m}{M} \bigr)^{-(P_l-1)}\,m^{-l}\,
\Bigl(f^{(m)}(x)-\sum_{k=0}^{l-1} m^k\,u_k^{(m)}(x)-m^l\,
\bigl( \log \tfrac{m}{M} \bigr)^{P_l}\,u_{l,P_l}(x)\Bigr)\notag\\
&=u_{l,P_l-1}\ ;
\end{align*}
and so on.

{\it Part (6)}: from property (E) we know that the distribution
$$
t^{(m)}(x):=m^{-(L+1)}\,\mathfrak{r}^{(m)}_{L+1}(x)\quad\text{fulfils}\quad
\lim_{m\downto 0}\,(\tfrac{m}{M})^\eps\,t^{(m)}=0\quad\quad\forall\eps>0\ .
$$
From (D) we conclude that 
\be
\rho^{D-(L+1)}\,t^{(m)}(\rho x)=t^{(\rho m)}(x)+
\sum_{k=1}^N l_k^{(\rho m)}(x)\,(\log\rho)^k \quad\quad\forall\rho >0
\ee
with some $l_k^{(m)}\in\sD^\prime(\bR^k[\less\{0\}])$. Multiplying the latter equation
by $(\rho m)^\eps$ and performing the limit $m\downto 0$, we conclude that 
$$
\lim_{m\downto 0}(\tfrac{m}{M})^\eps\,l_k^{(m)}=0\quad\quad\forall\eps>0\ ,\,\,\,k=1,...,N\ .
$$
It follows that
\begin{align*}
\lim_{\rho\downto 0}\rho^{D-(L+1)+\eps}\,\mathfrak{r}^{(m)}_{L+1}(\rho x)& =m^{L+1}\Bigl(
\lim_{\rho\downto 0}\rho^\eps\,t^{(\rho m)}(x)\\
+\sum_{k=1}^N& \bigl(\lim_{\rho\downto 0}\rho^{\eps/2}\,l_k^{(\rho m)}(x)\bigr)\,
\bigl(\lim_{\rho\downto 0}\rho^{\eps/2}\,(\log\rho)^k\bigr)\Bigr)=0\quad\forall\eps >0\ .
\end{align*}
\end{proof}

From parts {\it (1)} and  {\it (2)} we see that
the new axiom~(k), sm-expansion, is \textit{sufficient} for the 
above proposed axioms (i$'$), continuity in $m \geq 0$, and 
(g), almost homogeneous scaling. We will see that~(k) is even
\emph{equivalent} to the combination of~(i$'$) and (g), in the sense that the set of
solutions of the axioms (a)-(f), (h) and (k) is \emph{equal} to the set of solutions of
(a)-(h) and~(i$'$).

\section{Construction of a solution of the new system of axioms} 

In this section we use the inductive Epstein-Glaser construction \cite{EpsteinG73}, to obtain the 
general solution of the system of axioms (a)-(f), (h) and (k). More precisely we work with
Stora's extension of distributions \cite{Stora82,BrunettiF00} instead of Epstein and Glaser's
distribution splitting method.

\subsection{Inductive step, off the thin diagonal} 

We use that $T_n^0(A_1(x_1),...)\in\sD'(\bM^n\less\Dl_n,\sF)$ \eqref{eq:T-0}
is uniquely determined by causal factorization \eqref{eq:Tprods-causal}, see \cite{BrunettiF00}. 
Due to the uniqueness of the sm-expansion, we only have to show that
for every configuration $(x_1,...,x_n)\in\bM^n\setminus\Dl_n$ there {\bf exists} such an expansion; in
particular, the resulting expansion does not depend on the way we split $\{x_1,...,x_n\}$ into two
nonempty subsets such that one is later than the other.

Without restricting generality, we may assume that 
$\{x_1,...,x_l\}\cap $\newline $(\{x_{l+1},...,x_n\}+\bar V_-)=\emptyset\ $, 
in addition let $A_1,...,A_n$ be field {\it monomials}. 
Inserting the causal Wick expansion \eqref{eq:causal-Wick-expan} into \eqref{eq:Tprods-causal}, we see
that $t^0(A_1,...,A_n)(x_1-x_n,...):=\om_0\bigl(T^0_n(A_1(x_1),...)\bigr)$
is a linear combination of products
\begin{align}\label{eq:j-wick-exp}
&t(\unl A_1,...,\unl A_l)(x_1-x_l,...)\,
t(\unl A_{l+1},...,\unl A_n)(x_{l+1}-x_n,...)\notag\\
&\quad \quad \cdot \om_0\Bigl(\bigl(\ovl A_1(x_1)\cdots\ovl A_l(x_l)\bigr)\star_m
\bigl(\ovl A_{l+1}(x_{l+1})\cdots\ovl A_n(x_n)\bigr)\Bigr)\ .
\end{align}
The $\om_0(...)$-factor is, if it does not vanish, a linear combination of products
\be
\prod_{k=1}^K\del^{\beta_k} \Dl^+_m(x_{i_k}-x_{j_k})\quad\text{with}\quad
K(d-2)+\sum_{k=1}^K |\beta_k|=\sum_{i=1}^n\mathrm{dim}\,\ovl A_i\ ,
\ee
where $i_k\in \{1,...,l\}$ and $j_k\in \{l+1,...,n\}$.
By induction $t(\unl A_1,\unl A_l)$ and 
$t(\unl A_{l+1},\unl A_n)$ fulfil the sm-expansion with 
degree $D_{(i)}:=\sum_{i=1}^l\dim \unl A_i$ and 
$D_{(ii)}:=\sum_{j=l+1}^n\dim \unl A_j$, respectively; in addition
$\del^{\beta_k} \Dl^+_m$ satisfies this expansion with degree $D_k:=d-2+|\beta_k|$
(due to part {\it (3)} of the lemma). By means of part {\it (4)} of the lemma, we conclude that 
\eqref{eq:j-wick-exp} fulfils the sm-expansion with degree
$$
D_{(i)}+D_{(ii)}+\sum_{k=1}^K D_k=\sum_{i=1}^n\dim A_i\ ,
$$
where we use that $\dim\unl A+\dim\ovl A=\dim A$ (which follows immediately
from \eqref{subpolynomials}).
Hence, $T_n^0$ fulfils the new axiom~(k).

\subsection{Extension to the thin diagonal} 

To maintain the sm-expansion of $t_n^{(m)\,0}\in\sD'(\bR^{d(n-1)} \less \{0\})$,
\begin{align}\label{eq:s-expansion}
t_n^{(m)\,0}(x)& = u_0^0(x)+\sum_{l=1}^L m^l\, 
\sum_{p=0}^{P_l} \bigl( \log \tfrac{m}{M} \bigr)^p\,u^0_{l,p}(x)
+\mathfrak{r}_{L+1}^{(m)\,0}(x)\ ,
\end{align} 
we extend each distribution $u_0^0,\,u^0_{l,p},\,\mathfrak{r}_{L+1}^{(m)\,0}\in \sD'(\bR^{d(n-1)} \less \{0\})$  
individually.

Due to part {\it (6)} of the lemma, the remainders
$$
\mathfrak{r}_{L+1}^{(m)\,0}\quad\text{with}\quad L\geq L_0:=D-d(n-1)
$$
can be extended by the direct extension \eqref{eq:direct-extension}.

The distributions $u^0_{l,p}$ ($l\geq 1$) and $u^0_0$ ($l=0$) scale almost homogeneously 
in~$x$ with degrees $(D{-}l)$. Thus, by
proposition~\ref{pr:extens-exist}, there exist extensions
$u_{l,p} \in \sD'(\bR^{d(n-1)})$ and $u_0 \in \sD'(\bR^{d(n-1)})$, 
respectively, which scale almost homogeneously with the same degree as
the corresponding $u^0_{\cdots}$-distributions. For $l>L_0$
the almost homogeneous extension is unique and agrees with the direct 
extension \eqref{eq:direct-extension}. For $0\leq l\leq L_0$ the
extension needs a mass scale $M_1 > 0$; we choose $M_1$ independent 
of~$m$, such that $\del_mu_{l,p} = 0$ and $\del_mu_0 = 0$. One may 
choose $M_1=M$.  

We have to maintain the relation
\be\label{eq:r-relation}
\mathfrak{r}_{L_1+1}^{(m)\,0}(x)=\mathfrak{r}_{L_2+1}^{(m)\,0}(x)+
\sum_{l=L_1+1}^{L_2} m^l\,\sum_{p=0}^{P_l} \bigl( \log \tfrac{m}{M} \bigr)^p\,u^0_{l,p}(x)
\ ,\quad 0\leq L_1<L_2\ .
\ee
For $L_1\geq L_0$ the extensions indeed satisfy this relation, because all 
distributions appearing in \eqref{eq:r-relation} are extended by the unique 
direct extension \eqref{eq:direct-extension}. For $L_1<L_0$ we fulfil 
\eqref{eq:r-relation} by {\it defining}
the extension of $\mathfrak{r}_{L_1+1}^{(m)\,0}$ by
$$
\mathfrak{r}_{L_1+1}^{(m)}(x):=\mathfrak{r}_{L_0+1}^{(m)}(x)+
\sum_{l=L_1+1}^{L_0} m^l\,\sum_{p=0}^{P_l} \bigl( \log \tfrac{m}{M} \bigr)^p\,u_{l,p}(x) 
\ \quad\text{for}\quad 0\leq L_1<L_0\ .
$$

An extension~$t_n^{(m)}\in \sD'(\bR^{d(n-1)})$ of~$t_n^{(m)\,0}$, 
which fulfils the sm-expansion
(with the same degree $D$ as~$t_n^{(m)\,0}$), is obtained by 
inserting the constructed extensions of the various distributions
into \eqref{eq:s-expansion}; it does not matter which $L$ we use, since the 
extensions fulfil \eqref{eq:r-relation}.

From the preceding subsection we only know that $t^0(A_1,...,A_n)$ satisfies
the sm-expansion for field {\it monomials} $A_1,...,A_n$. Hence, we have to explain, 
how the just described construction matches with the procedure \eqref{eq:AWI-extension}
(in which the extension is done first for {\it balanced fields}). To explain this, note that,
due to linearity of the map $\otimes_{i=1}^n A_i\mapsto t^0(A_1,\ldots,A_n)$, 
the sm-expansion holds for $t^0(A_1,\ldots,A_n)$ for all
$A_1,\ldots,A_n\in\sP_\homog$ (and not only for field {\it monomials}).
With that an extension $t(A_1,\ldots,A_n)$ which fulfills the sm-expansion 
can be constructed as just described for all $A_1,\ldots,A_n\in\sP_\bal\cap\sP_\homog$.
Symmetrization w.r.t.~permutations of $(A_1,x_1),...,(A_n,x_n)$ does not violate the 
sm-expansion. Then, by means of \eqref{eq:AWI-extension}, we construct
$t(A_1,...,A_n)$ for all  $A_1,\ldots,A_n\in\sP$. To complete the inductive step,
we have to show that, on the level of the extensions, 
the sm-expansion holds for all {\it monomials}
$A_1,\ldots,A_n$ (and not only for $A_1,\ldots,A_n\in\sP_\bal\cap\sP_\homog$). 
For this purpose we write arbitrary monomials
$A_i$ ($1\leq i\leq n$) as $A_i=\sum_{k_i}\del^{\beta_{ik_i}}B_{ik_i}$ with
$B_{ik_i}\in\sP_\bal\cap\sP_\homog\ $. Note that 
$\dim  B_{ik_i}+|\beta_{ik_i}|=\dim A_i\ ,\,\,\forall k_i$.
Then, $t(A_1,...,A_n)$ is given in terms of the distributions $t(B_{1k_1},...,B_{nk_n})$
by \eqref{eq:AWI-extension}. In this formula, each summand fulfils the sm-expansion with degree
$$
\sum_{i=1}^n\dim B_{ik_i}+\sum_{i=1}^n|\beta_{ik_i}|=\sum_{i=1}^n\dim A_i\ ,
$$
hence, this holds also for $t(A_1,\ldots,A_n)$.

\vspace{6pt}

The {\bf most general solution} of the system of axioms is obtained by adding to a particular 
solution $t^{(m)}(A_1,...,A_n)(x_1-x_n,...)$ a polynomial in derivatives of the delta
distribution which fulfils the sm-expansion:
\be\label{eq:nonunique}
\sum m^l\,(\log\tfrac{m}{M})^p\, C_{l,p,\beta}(A_1,\dots,A_n) 
\,\del^\beta\dl(x_1 - x_n,\dots,x_{n-1} - x_n)\ ,
\ee
where the sum runs over $l \in \bN_0$, $p \in \bN_0$ and
$\beta \in \bN_0^{d(n-1)}$, with the restrictions
\be\label{eq:nonunique1}
|\beta| + l = D - d(n - 1)\quad\text{and}\quad  p \leq P \,\,\,\text{for some}\,\,\,P< \infty\ ;
\ee
the numbers $C_{l,p,\beta}(A_1,\dots, A_n)\in \bC$ do not depend on~$m$.
In addition \eqref{eq:nonunique} has to be Lorentz covariant and invariant under 
permutations of $(A_1,x_1),\ldots ,$\newline $(A_n,x_n)$; the coefficients $C_{l,p,\beta}(A_1,\dots, A_n)$
are also restricted by further axioms as e.g.~unitarity.

\vspace{6pt}

We return to the assertion at the end of sect.~3: if we replace the axiom (k) by the
(possibly weaker) axioms (g) and (i$'$), the freedom of (re)normalization
\eqref{eq:nonunique}-\eqref{eq:nonunique1} does not get bigger. (This follows from the 
discussion in \eqref{eq:nonunique-scaling}-\eqref{eq:continuous-m}.) Therefore, 
the two systems of axioms are indeed equivalent.

\section{The scaling and mass expansion for a dimensionally regularized theory}

In \cite{Keller13} dimensional regularization in position space is introduced by a 
change of the order of the Bessel functions defining the propagators: the regularized Feynman 
propagator is of the form
\be\label{eq:reg-Feyn-prop}
\Delta^{F\,\zeta}_m(x)=\sum_{l=0}^\infty h_l^\zeta\,
M^{2\zeta}\,m^{2l}\,(-(x^2-i\epsilon))^{l+1-\tfrac{d}2+\zeta}
+\sum_{l=0}^\infty c_l^\zeta\,M^{2\zeta}\, m^{d-2+2l-2\zeta}\,(-x^2)^l\ ,
\ee
where $\zeta\in\Omega\less\{0\}$ for a neigborhood $\Omega\subset\bC$ of $0$; 
and $M>0$ is a mass parameter, the factor
$M^{2\zeta}$ is introduced to keep the mass dimension constant. 
The coefficients $h_l^\zeta,\,c_l^\zeta\in\bC$ do not depend on $(x,m)$.
In the limit $\zeta\to 0$, $\Delta^{F\,\zeta}_m(x)$
converges in a suitable sense to $\Delta^{F}_m(x)$. From \eqref{eq:reg-Feyn-prop} we see
that $\Delta^{F\,\zeta}_m(x)$ is homogeneous under $(x,m)\to(\rho x,m/\rho)$:
\be\label{eq:reg-Feyn-prop-scaling}
\rho^{d-2-2\zeta}\,\Dl_{\rho^{-1}m}^{F\,\zeta}(\rho x)=\Dl_m^{F\,\zeta}(x)\ .
\ee 

To find the sm-expansion for the so regularized theory, we study a product of derivated, 
regularized Feynman propagators -- with different $\zeta_{ij}$ for different arguments $(x_i-x_j)$,
since the Epstein-Glaser forest formula requires the ability to vary the regularization 
parameters independently in this way, see \cite{Keller13}. We only treat the 
{\it even dimensional} case.\footnote{In odd dimensions, $(\tfrac{m}{M})^{2p-2{\bf c}\bzeta}$
is replaced by $(\tfrac{m}{M})^{p-2{\bf c}\bzeta}\ ,\,\,p\in\bN_0$.}
For $x_i\not= x_j\,\,\forall i<j$, we obtain the structure
\be\label{eq:prod-reg-Feyn-prop}
\prod_{k=1}^Q \partial^{\beta_k}\Delta^{\zeta_{i_kj_k}}_{F,\,m}(x_{i_k}-x_{j_k})=
\sum_{|{\bf c}|+|{\bf h}|=Q} \sum_{p=0}^\infty (\tfrac{m}{M})^{2p-2{\bf c}\bzeta}\,
u_{p,{\bf c},{\bf h}}^\bzeta(x)\ ,
\ee
where $x:=(x_1-x_n,...,x_{n-1}-x_n)$, and
$h_{ij}\in\bN_0$ ($c_{ij}\in\bN_0$ resp.) is the number of 
$h$-lines ($c$-lines resp.) (i.e.~the propagator
is given by a $h_l^\zeta$-term ($c_l^\zeta$-term resp.)) connecting the vertices $x_i$ and $x_j$, and
$$
\bzeta:=(\zeta_{ij})_{i<j}\ ,\quad {\bf c}:=(c_{ij})_{i<j}\ ,\quad |{\bf c}|:=\sum_{i<j}c_{ij}\ ,\quad
{\bf c}\bzeta:=\sum_{i<j}c_{ij}\zeta_{ij}\ ,
$$
and ${\bf h},\,|{\bf h}|$ and ${\bf h}\bzeta$ are similarly defined.
In addition the $m$-independent distributions $u_{p,{\bf c},{\bf h}}^\bzeta(x)$ are homogeneous: 
\be
\rho^{\kappa}\,u_{p,{\bf c},{\bf h}}^\bzeta(\rho x)=u_{p,{\bf c},{\bf h}}^\bzeta(x)
\quad\text{with}\quad \kappa:=Q(d-2)-2p-2{\bf h}\bzeta+\sum_k|\beta_k|\ .
\ee
It follows that on the r.h.s.~of \eqref{eq:prod-reg-Feyn-prop} the sum 
$\sum_p (\tfrac{m}{M})^{2p-2{\bf c}\bzeta}\,u_{p,{\bf c},{\bf h}}^\bzeta$
is homogeneous under $(x,m)\to(\rho x,m/\rho)$ with degree
$$ 
\kappa+2p-2{\bf c}\bzeta =Q(d-2)+\sum_k|\beta_k|-2\,({\bf h}+{\bf c})\bzeta\ .
$$

This motivates to require the following version of the sm-expansion axiom for the
$\bzeta$-dependent regularized time-ordered product $T^{(m)\,\bzeta}\equiv (T^{(m)\,\bzeta}_n)$: 
for a field monomial $A=\prod_{j=1}^J \partial^{\beta_j}\vf$ let $|A|:=J$ and, similarly to \eqref{eq:t},
we define the vacuum expectation values $t^{(m)\,\bzeta}(A_1,\dots,A_n)\in\sD'(\bR^{d(n-1)})$.
In addition let $N:= {n\choose 2}$.
\begin{itemize} 
\item \textbf{Scaling and mass expansion ($d>2$ even):} There exists an open neighborhood 
$\Omega_n\subset\bC^N$ of the origin such that
for all field {\it monomials} $A_1,\dots,A_n\in \sP$, 
the distributions $t^{(m)\,\bzeta}(A_1,\dots,A_n)(x_1 - x_n,\dots, x_{n-1} - x_n)$
fulfil for $\bzeta\in\Omega_n\less\{0\}$
the regularized sm-expansion with degree $D=\sum_{k=1}^n\dim A_k\in\bN_0$ and $l=
\tfrac{1}2\,\sum_{k=1}^n |A_k|\in\bN_0$ lines; where the following definition is used:

\begin{defn} 
\label{df:sm-expansion-reg}
Let $\Lambda\subset\bC^N$ be an open set.
A distribution $f^{(m)\,\bzeta}\in\sD'(\bR^{d(n-1)}[\setminus\{0\}])$, depending on $m\geq 0$, 
fulfils for $\bzeta\in\Lambda$ the {\it regularized sm-expansion} with degree $D$ and
$l\in\bN_0$ lines, if it is analytic in $\bzeta\in\Lambda$, and if
for all $p,P\in\bN_0$ and ${\bf c},\,{\bf h}\in\bN_0^N$ with $|{\bf c}|,|{\bf h}|\leq l$, 
there exist $m$-independent distributions $u_{p,{\bf c},{\bf h}}^\bzeta\in\sD'(\bR^{d(n-1)}[\setminus\{0\}])$
and remainders $\mathfrak{r}_{P+1,{\bf c},{\bf h}}^{(m)\,\bzeta}\in\sD'(\bR^{d(n-1)}[\setminus\{0\}])$, 
such that
\be\label{eq:sm-exp-reg}  
f^{(m)\,\bzeta}(x)=\sum_{|{\bf c}|+|{\bf h}|=l}\Bigl[
\sum_{p=0}^P (\tfrac{m}{M})^{2p-2{\bf c}\bzeta}\,u_{p,{\bf c},{\bf h}}^\bzeta(x)+
\mathfrak{r}_{P+1,{\bf c},{\bf h}}^{(m)\,\bzeta}(x)\Bigr]\ ,\quad\forall P\in\bN_0
\ee
and $\forall\bzeta\in\Lambda\ $; in addition
\begin{itemize}
\item[(A)] for $p=0$ and ${\bf c}\not={\bf 0}$ we have $u_{0,{\bf c},{\bf h}}^\bzeta\equiv 0\,\,
\forall {\bf h}\ $, and for $m=0$ it holds
$f^{(0)\,\bzeta}=\sum_{|{\bf h}|=l}u_{0,{\bf 0},{\bf h}}^\bzeta\ $;
\item[(B)] for ${\bf h}={\bf 0}$ we have $u_{p,{\bf c},{\bf 0}}^\bzeta\in C^\infty\ $;
\item[(C)] $u_{p,{\bf c},{\bf h}}^\bzeta(x)$ is {\it homogeneous} 
(not only almost homogeneous) in $x$ with degree
\be
\kappa^\bzeta_{p,{\bf h}}:=D-2p-2{\bf h}\bzeta\ ;
\ee
\item[(D)] $\mathfrak{r}_{P+1,{\bf c},{\bf h}}^{(m)\,\bzeta}(x)$ is homogeneous 
under $(x,m)\to(\rho x,m/\rho)$ with degree
\be\label{eq:D-zeta-c-h}
D^\bzeta_{{\bf c},{\bf h}}=D-({\bf h}+{\bf c})\bzeta\ ;
\ee
\item[(E)] $\mathfrak{r}_{P+1,{\bf c},{\bf h}}^{(m)\,\bzeta}(x)$ is smooth in $m$ for $m>0$ and
\be
\lim_{m\downarrow 0} (\tfrac{m}{M})^{-2(P+1)+2{\bf c}\bzeta+\epsilon}\,\,
\mathfrak{r}_{P+1,{\bf c},{\bf h}}^{(m)\,\bzeta}=0
\quad\quad\forall\epsilon >0\ .
\ee
\end{itemize}
\end{defn}
\end{itemize}
Similarly to (C) and (D), the properties (A) and (B) are motivated 
by their validity for \eqref{eq:prod-reg-Feyn-prop}.
(B) is important for the extension of the distributions 
$u_{p,{\bf c},{\bf h}}^{\bzeta\,0}\in\sD'(\bR^{d(n-1)}\less\{0\})$:
for almost all values of $\bzeta\in\Lambda$ we have 
$\kappa^\bzeta\not\in d(n-1)+\bN_0$ (i.e.~we are in the 
much simpler case (i) of proposition \ref{pr:extens-exist}). 

Suitably modified, all statements of lemma \ref{lm:sm-expans-properties} 
hold true also for the regularized 
sm-expansion. The modifications are:\footnote{For shortness we do not specify the 
domain for $\bzeta$.} 
let $(D,l)$ be the degree and the 
number of lines in the regularized 
sm-expansion of the distribution $f^{(m)\,\bzeta}\in\sD'(\bR^{d(n-1)}[\setminus\{0\}])$.
\begin{itemize}
\item[{\it (1$'$)}] (No change for $m>0$.) In order that the limit $m\downarrow 0$ exists,
we assume that $\Re(\zeta_{ij})<\tfrac{1}{l}\,\,\forall i,j$ 
(which implies $\Re({\bf c}\bzeta)<1$). With that it holds
\be
\lim_{m\downarrow 0}f^{(m)\,\bzeta} =\sum_{|{\bf h}|=l} u^\bzeta_{0,{\bf 0},{\bf h}}=
f^{(0)\,\bzeta}\ .
\ee
\item[{\it (2$\,'$)}] 
Only the expression in the $\bigl[...\bigr]$-bracket of \eqref{eq:sm-exp-reg} 
(and not the complete $f^{(m)\,\bzeta}$)
is {\it homogeneous} under $(x,m)\to(\rho x,m/\rho)$, with 
degree $D^\bzeta_{{\bf c},{\bf h}}$ \eqref{eq:D-zeta-c-h}.
\item[{\it (3$\,'$)}] $\del_x^\beta f^{(m)\,\bzeta}(x)$ fulfils the 
regularized sm-expansion with $(D+|\beta|,l)$.
\item[{\it (4$'$)}] We formulate the statement in the form in which 
it is used in the inductive step of the construction of 
$T^{(m)\,\bzeta}$: let $\Delta^{+\,\zeta}_m$ be the regularized two-point 
function belonging to $\Delta^{F\,\zeta}_m\ $.\footnote{That is 
$\Delta^{F\,\zeta}_m(x)=\theta(x^0)\Delta^{+\,\zeta}_m(x)+\theta(-x^0)\Delta^{+\,\zeta}_m(-x)\ $.} 
We assume that
 $f_1^{(m)\,\bzeta_1}(x_1-x_s,...)\in\sD'(\bR^{d(s-1)})$ and  
$f_2^{(m)\,\bzeta_2}(x_{s+1}-x_n,...)\in\sD'(\bR^{d(n-s-1)})$
fulfil the regularized sm-expansion with $(D_1,l_1)$ and $(D_2,l_2)$, respectively. Then,
\be
f_1^{(m)\,\bzeta_1}(x_1-x_s,...)\, f_2^{(m)\,\bzeta_2}(x_{s+1}-x_n,...)\,
\prod_{k=1}^K\del^{\beta_k} \Dl^{+\,\zeta_{i_kj_k}}_m(x_{i_k}-x_{j_k})
\ee
(where $i_k\in \{1,...,s\}$ and $j_k\in \{s+1,...,n\}\,\,\forall k$),
which is an element of  $\sD'(\bR^{d(n-1)}\setminus\{0\})$,
satisfies  the regularized sm-expansion with
\be
\bzeta:=\bigl(\bzeta_1,\bzeta_2,(\zeta_{ij})_{i\in \{1,...,s\}}^{j\in \{s+1,...,n\}}\bigr)\ ,\quad
D=D_1+D_2+K(d-2)+\sum_{k=1}^K|\beta_k|
\ee
and $l=l_1+l_2+K\ $.
\item[{\it (5$\,'$)}] If we know that a given $f^{(m)\,\bzeta}$ fulfils the regularized sm-expansion with 
given numbers $(D,l)$, then the coefficients
$u_{p,{\bf c},{\bf h}}^\bzeta$ are uniquely determined.
\item[{\it (6$\,'$)}] $\sd(\mathfrak{r}_{P+1,{\bf c},{\bf h}}^{(m)\,\bzeta})\leq \Re(\kappa^\bzeta_{P+1,{\bf h}})=D-2(P+1)-2\Re({\bf h}\bzeta)\ $.
\end{itemize}

\begin{proof} {\it (1$'$)}, {\it (2$\,'$)} and {\it (6$\,'$)} are easy. 
(Note that {\it (2$\,'$)} and {\it (6$\,'$)}
are simpler to prove than the corresponding statements in  lemma \ref{lm:sm-expans-properties},
since $u_{p,{\bf c},{\bf h}}^\bzeta$ and $\mathfrak{r}_{P+1,{\bf c},{\bf h}}^{(m)\,\bzeta}$ scale even {\it homogeneously}.)

{\it (3$\,'$)} can be verified in the same way as in  lemma \ref{lm:sm-expans-properties}.

{\it (4$'$)} can be proved by proceeding analogously to the unregularized theory (see part {\it (4)} of 
lemma \ref{lm:sm-expans-properties} and sect.~4.1) and by using that $\Delta^{+\,\zeta}_m$ is
also of the form \eqref{eq:reg-Feyn-prop} (one only has to replace $(x^2-i\eps)$ by $(x^2-ix^0\eps)$).

To prove  {\it (5$\,'$)} let $\bzeta\in\Lambda$ be such that
\be
p-\Re({\bf c}\bzeta)\not=p'-\Re({\bf c'}\bzeta) \quad\forall (p,{\bf c})\not=(p',{\bf c'})
\,\quad\text{and}\,\quad
{\bf h}\bzeta\not={\bf h'}\bzeta \quad\forall {\bf h}\not={\bf h'}\ .
\ee
This excludes only a set of measure zero -- this is no harm, due to analyticity in $\bzeta$. 
The first condition implies that $f^{(m)\,\bzeta}$ is of the form
\be
f^{(m)\,\bzeta}=\sum_{i=1}^K U_i\,(\tfrac{m}{M})^{z_i}+
\mathfrak{r}^{(m)}_{K+1}\quad\text{with}\quad \Re(z_i)<\Re(z_{i+1})\quad\forall i
\ee
and $\lim_{m\downarrow 0} (\tfrac{m}{M})^{-z_K}\,\mathfrak{r}_{K+1}^{(m)}=0\ $,
where $K\in\bN$ is arbitrary.
The coefficients $U_i$ can be determined inductively: 
\be
U_n=\lim_{m\downarrow 0}\Bigl(f^{(m)\,\bzeta}-\sum_{i=1}^{n-1} U_i\,(\tfrac{m}{M})^{z_i}\Bigr)\,
(\tfrac{m}{M})^{-z_n}\ .
\ee
Finally from $U_i=\sum_{{\bf h}} u_{p,{\bf c},{\bf h}}^\bzeta$, where 
$z_i=2(p-{\bf c}\bzeta)$ and the sum is restricted by
$|{\bf h}|=l-|{\bf c}|$, a single summand is obtained by the projection
\be
u_{p,{\bf c},{\bf h}_0}^\bzeta=\frac{\prod_{{\bf h}\not={\bf h}_0}
\Bigl(D-2p-2{\bf h}\bzeta+\sum_r x_r\partial_{x_r}\Bigr)}
{\prod_{{\bf h}\not={\bf h}_0}2\,({\bf h}_0-{\bf h})\bzeta}\,\,U_i\ .
\ee
\end{proof}

Notice that for $f^{(m)\,\bzeta}=t^{(m)\,\bzeta}(A_1,...,A_n)$ (where $A_1,...,A_n$ are arbitrary field monomials)
the property {\it (2$\,'$)} is an equivalent formulation of the axiom 'Scaling' in \cite{Keller13}.

The system of axioms for the regularized time-ordered product $T^{(m)\,\bzeta}$ given
in \cite{Keller13} can now be modified as follows: similarly to the procedure in sect.~3, we replace the 
axioms 'Smoothness in $m^2$' and 'Scaling' by the sm-expansion axiom. Essentially by the same construction as 
in sect.~4, one obtains the general solution of the so modified system of axioms.

\section{Applications of the scaling and mass expansion}

The sm-expansion is very helpful for {\it practical computations}: 
choosing $L=L_0=D-d(n-1)$ it reduces the main problem -- the extension 
from $\sD'(\bR^{d(n-1)} \less \{0\})$ to $\sD'(\bR^{d(n-1)})$ --
to a minimal set of almost homogeneous scaling distributions
(namely $\{u^0_{l,p}\,|\,0\leq l\leq L_0,\,0\leq p\leq P_l\}$);
the direct extension \eqref{eq:direct-extension} of the remainder gives no computational work.
We illustrate this by the following examples.

\begin{ex}[setting sun diagram]
We study again the setting sun diagram in $d=4$ dimensions.
We have to extend
\be\label{eq:settingsun-unren}
t^0(\vf^3,\vf^3)(x)=(\Dl^F_m(x))^3\in\sD'(\bR^4 \less \{0\})\ ,
\ee
where $\Dl^F_m$ is the Feynman propagator. 
Due to \eqref{eq:Delta-plus} its sm-expansion can be written as
\be\label{eq:sm-Feynprop}
\Dl^F_m(x)=\frac{a_0}{X}+m^2\,\Bigl(\bigl(a_1\,\log(M^2X) 
+ A_1\bigr)+2a_1\,\log\tfrac{m}{M}\Bigr)+R_4^{(m)}(x)\ ,
\ee
where $X:=-(x^2 - i 0)\ $, with constants $a_0,\,a_1,\,A_1\in\bC$.
Due to $D=6,\,n=2$, we have $L_0=2$. Using that, we
insert \eqref{eq:sm-Feynprop} into \eqref{eq:settingsun-unren} and obtain
$$
t^0(\vf^3,\vf^3)(x)=u_0^0(x)+m^2\bigl(u^0_{2,0}(x)+u^0_{2,1}(x)\,
\log\tfrac{m}{M}\bigr)+\mathfrak{r}_4^{(m)\,0}(x)\ ,
$$
where 
\begin{align*}
u^0_0(x)=&\tfrac{a_0^3}{X^3}\ ,\quad 
u^0_{2,0}(x)=\tfrac{3\,a_0^2\,(a_1\,\log(M^2X) + A_1)}{X^2}\ ,\quad 
u^0_{2,1}(x)=\tfrac{6\,a_0^2\,a_1}{X^2}\ ,\\
\mathfrak{r}_4^{(m)}(x)=& 3\,R_4^{(m)}(x)\,(\Dl^F_m(x))^2
+3\,m^4\,(a_1\,\log(m^2X) + A_1)^2\,\tfrac{a_0}{X}\\
&+m^6\,(a_1\,\log(m^2X) + A_1)^3\ .
\end{align*}
Note that $u^0_{2l+1}=0\,\,\forall l\in\bN$.

The non-direct, almost homogeneous extensions of $u^0_0(x),\, u^0_{2,0}$ and $u^0_{2,1}$ 
can be computed by using differential renormalization 
(see e.g.~\cite[Appendix B]{DuetschF04} and references cited there) -- we use $M_1=M$ as renormalization
mass scale:
\begin{align}\label{eq:diffren}
u_0(x)&=a_0^3\,\square_x\square_x\ovl{\Bigl(\tfrac{\log(M^2X)}{32\,X}\Bigr)}
+C\,\square_x\delta(x) ,\notag\\ 
u_{2,0}(x)&=3\,a_0^2\,\Bigl[a_1\,\square_x\ovl{\Bigl(\tfrac{(\log(M^2X))^2+2\,\log(M^2X)}{8\,X}\Bigr)}
+A_1\, \square_x\ovl{\Bigl(\tfrac{\log(M^2X)}{4\,X}\Bigr)}\Bigr]
+C_0\,\delta(x)\ ,\notag\\ 
u_{2,1}(x)&=6\,a_0^2\,a_1\, \square_x\ovl{\Bigl(\tfrac{\log(M^2X)}{4\,X}\Bigr)}+C_1\,\delta(x)\ ,
\end{align}
where $C,\,C_0,\,C_1\in\bC$ are arbitrary constants.
These formulas have to be understood as follows: for $x\not= 0$ the derivatives can straightforwardly 
be computed and we obtain the corresponding $u^0_{...}$-distributions. However, the expressions in 
$\bigl(...\bigr)$-brackets have scaling degree $=2$, hence, by the direct extension
\eqref{eq:direct-extension} (denoted by an over-line), they are 
uniquely defined as elements of $\sD'(\bR^4)$, and also their derivatives are in $\sD'(\bR^4)$. 
Therefore, the r.h.~sides of \eqref{eq:diffren} are indeed extensions
of the corresponding $u^0_{...}$-distributions; and, obviously, they scale almost homogeneously. 

We end up with
\be\label{eq:settingsun-ren}
t(\vf^3,\vf^3)(x)=u_0(x)+m^2\bigl(u_{2,0}(x)+u_{2,1}(x)\,
\log\tfrac{m}{M}\bigr)+\mathfrak{r}_4^{(m)}(x)\in\sD'(\bR^4)\ ,
\ee
where $\mathfrak{r}_4^{(m)}$ is the direct extension of $\mathfrak{r}_4^{(m)\,0}$.
\end{ex}

\begin{ex}[setting sun with a hat]\label{exmp:settingsunhat}
Again in $d=4$ dimensions, we compute the ``divergent'' diagram
$$
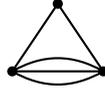
\begin{tikzpicture}[thick,scale=1.5]
\useasboundingbox (0,-0.1) rectangle (1,0.6);
\filldraw (0,0) circle (1pt);
\filldraw (0.4,0.6) circle (1pt);
\filldraw (0.8,0) circle (1pt);
\draw (0,0) edge [out=30,in=150]  (0.8,0);
\draw (0,0) --  (0.8,0);
\draw (0,0) edge [out=-30,in=-150] (0.8,0);
\draw (0,0) -- (0.4,0.6);
\draw (0.4,0.6) -- (0.8,0);
\end{tikzpicture} 
$$
which contains the setting sun diagram as a ``divergent'' subdiagram.\footnote{A
diagram with $n$ vertices is ``divergent'', iff its scaling degree \eqref{eq:sd-def}
is greater or equal
to $d(n-1)$, i.e.~the direct extension \eqref{eq:direct-extension} does not apply.}
That is we have to extend
\be\label{eq:triangle-subdiv}
t^0(x,y)=t(\vf^3,\vf^3)(x-y)\,\Dl^F_m(x)\,\Dl^F_m(y)\in\sD'(\bR^8\less\{0\}) \ ,
\ee
to $\sD'(\bR^8)$, where $t(\vf^3,\vf^3)$ is given by 
\eqref{eq:settingsun-ren}. We have $D=10,\,n=3$ 
and, hence, $L_0=2$. The sm-expansion of $t^0(x,y)$ with $L=L_0=2$ is obtained by inserting
\eqref{eq:sm-Feynprop} and \eqref{eq:settingsun-ren} into \eqref{eq:triangle-subdiv}:
$$
t^0(x,y)=v_0^0(x,y)+m^2\bigl(v^0_{2,0}(x,y)+v^0_{2,1}(x,y)\,
\log\tfrac{m}{M}\bigr)+\mathfrak{q}_4^{(m)\,0}(x,y)\ ,
$$
where we use the letters $(v,\mathfrak{q})$ (instead of $(u,\mathfrak{r})$)
to avoid confusion with the distributions appearing in the sm-expansion of the
setting sun diagram. The $v^0_{...}$-distributions read:
\begin{align*}
v^0_0(x,y)&=u_0(x-y)\,\tfrac{a_0^2}{XY}\ ,\\ 
v^0_{2,0}(x,y)&=u_{2,0}(x-y)\,\tfrac{a_0^2}{XY}+
u_0(x-y)\,a_0\Bigl(\tfrac{a_1\,\log(M^2Y) + A_1}{X}+\tfrac{a_1\,\log(M^2X) + A_1}{Y}\Bigr)\ ,\\
v^0_{2,1}(x,y)&=u_{2,1}(x-y)\,\tfrac{a_0^2}{XY}+
u_0(x-y)\,2a_0a_1\Bigl(\tfrac{1}{X}+\tfrac{1}{Y}\Bigr)\ ,
\end{align*}
where $Y$ is defined analogously to $X$ \eqref{eq:sm-Feynprop}.

Due to the choice $L=L_0=2$, the direct extension applies to the remainder 
$\mathfrak{q}_4^{(m)\,0}(x,y)$. 
The almost homogeneous extension of the $v^0_{...}$-distributions is more involved, we use
an {\it analytic regularization} which respects the ($x\leftrightarrow y$)-symmetry,
it is related to the methods in \cite{Carme,Bettina,NST11,Keller13} and 
\cite[Sect.3.4]{Hollands08}:
\be
v^{\zeta\,0}(x,y):=v^0(x,y)\,(M^4XY)^\zeta\ ,\quad\quad v=v_0,\,v_{2,0},\,v_{2,1}\ ,
\ee
where $\zeta\in\bC\less\{0\}$, $|\zeta|$ sufficiently small. The factor $M^{4\zeta}$ is introduced for
dimensional reasons. 

For a general $\zeta$, also $v^{\zeta\,0}$ cannot be renormalized by the direct extension. However,
we gain by the regularization that $v^{\zeta\,0}$ scales almost homogeneously with a 
{\it non-integer degree} $D^\zeta=8-4\zeta$ (for $v^{\zeta\,0}_{2,0},\,v^{\zeta\,0}_{2,1}$) or
$D^\zeta=10-4\zeta$ (for $v^{\zeta\,0}_0$). Due to that, the almost homogeneous extension 
$v^\zeta(x,y)$ is unique (proposition \ref{pr:extens-exist}) and can be computed by differential 
renormalization as follows:\footnote{For $v_{2,1}^\zeta$
and $v_{2,0}^\zeta$ we use the extension method given in \cite[remark 4.9]{Keller13}, for
$v_0^\zeta$ we work with a further development of that method.}
writing $z:=(x,y)$, $\del_r z_r:=\del_{x^\mu}x^\mu+\del_{y^\mu}y^\mu$
and $\eta:=-4\zeta$, we obtain from
$$
(\del_r z_r+\eta)^2 v^{\zeta\,0}_{2,1}(z)=0
$$
the unique almost homogeneous extension
\be\label{eq:ext-21}
v^{\zeta}_{2,1}=\tfrac{-1}{\eta^2}\,\Bigl((2\eta-1)\,\del_r \ovl{(z_r\, v^{\zeta\,0}_{2,1})}
+\del_r \del_s\ovl{(z_rz_s\, v^{\zeta\,0}_{2,1})}\Bigr)\in\sD'(\bR^8)\ .
\ee
Again, the over-line denotes the direct extension \eqref{eq:direct-extension}, which exists since 
$\sd(z_{r_1}...z_{r_l}\,v^{\zeta\,0})=D^\zeta-l$.
For  $v^{\zeta\,0}_{2,0}$ the power of the almost homogeneous scaling is $2$, hence we have
$$
(\del_r z_r+\eta)^3 v^{\zeta\,0}_{2,0}(z)=0\ ,
$$
which yields
\begin{align}\label{eq:ext-20}
v^{\zeta}_{2,0}=&\tfrac{-1}{\eta^3}\,\Bigl((3\eta^2-3\eta+1)\,\del_r \ovl{(z_r\, v^{\zeta\,0}_{2,0})}
+(3\eta-3)\,\del_r \del_s\ovl{(z_rz_s\, v^{\zeta\,0}_{2,0})}\notag\\
&+\del_p \del_r \del_s\ovl{(z_pz_rz_s\, v^{\zeta\,0}_{2,0})}\Bigr)\ .
\end{align}
For $v^{\zeta\,0}_0$ we need at least $l=3$ factors $z_{r_i}$ in order that the direct extension
$\ovl{z_{r_1}...z_{r_l}\, v^{\zeta\,0}_0}$ exists. Hence, we proceed as follows: from
$$
(\del_r z_r+2+\eta)^2 v^{\zeta\,0}_0(z)=0
$$
we obtain
$$
v^{\zeta\,0}_0=\tfrac{-1}{(2+\eta)^2}\,\Bigl((3+2\eta)\,\del_s (z_s\, v^{\zeta\,0}_0)
+\del_r \del_s (z_rz_s\, v^{\zeta\,0}_0)\Bigr)\ ,
$$
analogously
$$
(\del_r z_r+1+\eta)^2 (z_s\,v^{\zeta\,0}_0(z))=0
$$
gives
$$
z_s\,v^{\zeta\,0}_0=\tfrac{-1}{(1+\eta)^2}\,\Bigl((1+2\eta)\,\del_r (z_rz_s\, v^{\zeta\,0}_0)
+\del_p\del_r (z_pz_rz_s\, v^{\zeta\,0}_0)\Bigr)\ ,
$$
and
$$
(\del_p z_p+\eta)^2 (z_rz_s\,v^{\zeta\,0}_0(z))=0
$$
yields
$$
z_rz_s\,v^{\zeta\,0}_0=\tfrac{-1}{\eta^2}\,\Bigl((2\eta-1)\,\del_p (z_pz_rz_s\, v^{\zeta\,0}_0)
+\del_p\del_q (z_pz_qz_rz_s\, v^{\zeta\,0}_0)\Bigr)\ .
$$
Inserting the lower equations into the upper ones and performing the direct extension
we get
\begin{align}\label{eq:ext-0}
v^{\zeta}_0=\tfrac{1}{\eta^2(1+\eta)^2(2+\eta)^2}\,\Bigl(&(2+2\eta-6\eta^2-4\eta^3)\,
\del_p\del_r\del_s \ovl{(z_pz_rz_s\, v^{\zeta\,0}_0)}\notag\\
&-(2+6\eta+3\eta^2)\,\del_p\del_q\del_r\del_s \ovl{(z_pz_qz_rz_s\, v^{\zeta\,0}_0)}\Bigr)\ .
\end{align}
Obviously, the extensions $v^\zeta$ scale almost homogeneously with the same degree $D^\zeta$
and the same power as the initial $v^{\zeta\,0}$ (in agreement with proposition \ref{pr:extens-exist});
in addition, the maps $\zeta\mapsto \langle v^\zeta,f\rangle$ are meromorphic in $\zeta$ for all
$f\in\sD(\bR^8)$, with a pole at $\zeta=0$ of order $2$ (for $v^{\zeta}_{2,1},\,v^{\zeta}_0$) or
$3$ (for $v^{\zeta}_{2,0}$). The latter shows explicitly that this extension method does not work for 
the unregularized theory (i.e.~$\zeta= 0$).

According to definition 4.2 in \cite{Keller13}, $v^\zeta\in\sD'(\bR^8)$ is a 'regularization' of
$v^0\in\sD'(\bR^8\less\{0\})$ in the sense that
\be\label{eq:regularisation}
\lim_{\zeta\to 0}\langle v^\zeta,g\rangle=\langle v_\omega,g\rangle
\quad\quad\forall g\in\sD_\omega(\bR^8)\ ,
\ee
where $v_\omega$ is the unique extension of $v^0$ to $\sD_\omega'(\bR^8)$ \eqref{eq:Domega}
with $\sd(v_\omega)=\sd(v^0)$; and $\omega=0$ (for $v^{0}_{2,0},\,v^{0}_{2,1}$) or
$\omega=2$ (for $v^{0}_0$). Namely, using the functions $\chi_\rho$ \eqref{eq:direct-extension}
and that $\lim_{\zeta\to 0} v^{\zeta\,0}=v^0$ in $\sD'(\bR^8\less\{0\})$, 
\eqref{eq:regularisation} can be verified as follows: 
\begin{align}
\langle v_\omega,g\rangle&=\lim_{\rho\to\infty}\langle v^0,\chi_\rho \,g\rangle=
\lim_{\rho\to\infty}\,\lim_{\zeta\to 0}\langle v^{\zeta\,0},\chi_\rho \,g\rangle\notag\\
&=\lim_{\zeta\to 0}\,\lim_{\rho\to\infty}\langle v^{\zeta\,0},\chi_\rho \,g\rangle
=\lim_{\zeta\to 0}\langle v^{\zeta},g\rangle\ .
\end{align}

Turning to the limit $\zeta\to 0$, Corollary 4.4 in \cite{Keller13} states that the minimally 
subtracted distribution
\be\label{eq:MS}
v^\MS:=\lim_{\zeta\to 0}\,(1-\pp)\,v^\zeta
\ee
($\pp$ denotes the principle part) is an extension of $v^0$ with $\sd(v^\MS)=\sd(v^0)$. 

Coming back to the explicit Laurent series $v^\zeta=\sum_{n=-L}^\infty\zeta^n\,v_{(n)}$ 
(where $L\in\bN$) of our example, 
we have to compute the coefficients $v_{(0)}=v^\MS$. Expanding (in $\zeta$) $(M^4XY)^\zeta$ 
and the rational functions of $\eta$
appearing in \eqref{eq:ext-21}, \eqref{eq:ext-20} and \eqref{eq:ext-0}, we obtain the 
following results for the general, almost homogeneous and Lorentz invariant extensions
$v=v^\MS+\sum_{|\beta|=\omega}C_\beta\,\del^\beta\delta$, which are $(x\leftrightarrow y)$-invariant:
\begin{align}\label{eq:v-final}
v_{2,1}=&\del_r \ovl{\bigl(z_r\, v^0_{2,1}\,[\tfrac{1}{32}\,(\log(M^4XY))^2
+\tfrac{1}2\,\log(M^4XY)]\bigr)}\notag\\
&-\del_r \del_s\ovl{\bigl(z_rz_s\, v^0_{2,1}\,\tfrac{1}{32}\,(\log(M^4XY))^2\bigr)}
+C_1\,\delta(x,y)\ ,\notag\\
v_{2,0}=&\del_r \ovl{\bigl(z_r\, v^0_{2,0}\,[\tfrac{(\log(M^4XY))^3}{384}+
\tfrac{3\,(\log(M^4XY))^2}{32}+\tfrac{3\,\log(M^4XY)}{4}]\bigr)}\notag\\
&-\del_r \del_s\ovl{\bigl(z_rz_s\, v^0_{2,0}\,[\tfrac{3\,(\log(M^4XY))^3}{384}+
\tfrac{3\,(\log(M^4XY))^2}{32}]\bigr)}\notag\\
&+\del_p\del_r \del_s\ovl{\bigl(z_pz_rz_s\, v^0_{2,0}\,\tfrac{(\log(M^4XY))^3}{384}\bigr)}
+C_0\,\delta(x,y)\ ,\notag\\
v_0=&\del_p\del_r\del_s \ovl{\bigl(z_pz_rz_s\, v^0_0\,[\tfrac{-1}8
+\tfrac{1}4\,\log(M^4XY)+\tfrac{1}{64}\,(\log(M^4XY))^2]\bigr)}\notag\\
&+\del_q\del_p\del_r \del_s\ovl{\bigl(z_qz_pz_rz_s\, v^0_0\,[\tfrac{7}8
-\tfrac{1}{64}\,(\log(M^4XY))^2]\bigr)}\notag\\
&+C_2\,(\square_x+\square_y)\delta(x,y)+C_3\,\del^x_\mu\del_y^\mu\,\delta(x,y)\ .
\end{align}
We explicitly see that these extensions scale almost homogeneously with the same degree as the 
pertinent $v^0$-distributions. From proposition \ref{pr:extens-exist} we know that the power 
of the $\log$'s may be increased at most by $1$; therefore, terms of higher orders in 
$\log(M^2X)$, $\log(M^2Y)$ and $\log(M^2(X-Y))$ must cancel out in \eqref{eq:v-final},
by identities for the derivatives.
\end{ex}

\begin{rem}[treatment of subdivergences] There is an essential 
difference between the renormalization method used in this example and the one 
given in \cite{Keller13}: we insert for the divergent subdiagram (i.e.~the setting sun)
the {\it renormalized} expression and, 
hence, in the limit $\zeta\to 0$ we have to care only about the overall divergence located on the 
thin diagonal $x=0=y$. According to the method in \cite{Keller13}, one inserts for the  
divergent subdiagram a {\it regularized} expression and, therefore, the limit which removes the 
regularization has to be done by means of the forest formula: one first subtracts the principle part
of the divergent subdiagram (which is localized on the partial diagonal $x-y=0$) and, after that, one 
subtracts the principle part of the overall diagram (which is localized on the thin diagonal).
\end{rem}

\section{Concluding remarks}

In most papers dealing with causal perturbation theory (in particular in the original work 
\cite{EpsteinG73}) the scaling degree axiom (shortly 'sd-axiom') is used, which restricts
extensions $t\in\sD'(\bR^{d(n-1)})$ of $t^0\in\sD'(\bR^{d(n-1)}\less\{0\})$ by the requirement
$\sd(t)=\sd(t^0)$. In the system of axioms proposed by this paper (see sects.~3 and 4) one may
replace the sm-expansion axiom by the weaker sd-axiom -- this yields a reasonable system of axioms.

To illustrate that the sm-expansion axiom restricts the set of allowed time-ordered products 
truly stronger,
we discuss the non-uniqueness of the inductive step $n=2\,\to\,n=3$ for 
the example 'setting sun with a hat':
taking also Lorentz invariance and $(x\leftrightarrow y)$-symmetry 
into account, the sd-axiom leaves the freedom to add a term of the form
\be
\Bigl(f_2(\tfrac{m}{M})\,(\square_x+\square_y)+f_3(\tfrac{m}{M})\,\del_\mu^x\del^\mu_y+
m^2\,f_1(\tfrac{m}{M})\Bigr)\,\delta(x,y)\ ,
\ee
where $M>0$ is a fixed mass scale and $f_1,\,f_2,\,f_3$ are arbitrary functions
$f_i\,:\,\bR\to\bC$ (the values are dimensionless). We have found that
the sm-expansion axiom restricts these functions to
\be
f_2(\tfrac{m}{M})=C_2\ ,\quad f_3(\tfrac{m}{M})= C_3\ ,\quad
f_1(\tfrac{m}{M})=C_0+C_1\,\log(\tfrac{m}{M})\ ,
\ee
with arbitrary constants $C_0,\,C_1,\,C_2,\,C_3\in\bC$.

Such a reduction of the freedom of (re)normalization by a 
refinement of the sd-axiom is certainly desirable.
As explained in \eqref{eq:nonunique-scaling}, almost homogeneous scaling (axiom (g)) does not suffice, 
it needs to be supplemented, or replaced by a stronger condition. 
In \cite{DuetschF04} this problem is solved by quantizing with a Hadamard function and requiring
as an additional axiom smoothness in $m\geq 0$. For time ordered products based on the Wightman
two-point function, we have shown that the sm-expansion axiom is well suited for a stronger 
version of the sd-axiom. 

As an outlook we mention that the sm-expansion axiom can be used to derive structural results about
the renormalization group flow, see \cite{Duetsch-rg-ssb}.

\appendix

\section{Extension of distributions from $\sD'(\bR^k \less \{0\})$ to $\sD'(\bR^k)$}
 
For the convenience of the reader we recall some main results about the 
extension of a given distribution $t^0\in\sD'(\bR^k \less \{0\})$ to $t\in\sD'(\bR^k)$, 
proofs are given e.g.~in \cite{BrunettiF00,DuetschF04}.

Steinmann's \textit{scaling degree} \cite{Steinmann71} of a distribution $f \in \sD'(\bR^k)$ or
$f \in \sD'(\bR^k \less \{0\})$ is defined by
\be\label{eq:sd-def}
\sd(f)
:= \inf\set{r \in \bR : \lim_{\rho\downto 0} \rho^r\,f(\rho x) = 0}\ .
\ee

Let $\om:=\sd(t^0)-k$ and introduce the subspace of test functions
\be\label{eq:Domega}
\sD_\om(\bR^k) := \set{h \in \sD(\bR^k)
: \del^\beta h(0) = 0 \text{ for } |\beta| \leq \om}\ .
\ee
Then, $t^0$ has a {\it unique} extension $t_\om$ to $\sD'_\om(\bR^k)$ satisfying
the condition $\sd(t_\om) = \sd(t^0)$. $t_\om$ is called the 'direct extension',
it can be obtained by the limit
\be\label{eq:direct-extension}
\duo< t_\om, h> := \lim_{\rho\to\infty} \duo< t^0, \chi_\rho h>\ ,\quad \quad 
h\in\sD_\om(\bR^k)\ ,
\ee 
where $\chi_\rho(x):=\chi(\rho x)$ and  $\chi \in C^\infty(\bR^k)$ is such that $0 \leq \chi(x) \leq 1$, 
$\chi(x) = 0$ for $|x| \leq 1$ and $\chi(x) = 1$ for $|x| \geq 2$.

In particular, for $\sd(t^0)<k$, the extension $t\in\sD'(\bR^k)$ is uniquely fixed by the requirement 
$\sd(t) = \sd(t^0)$ and it is given by the direct extension \eqref{eq:direct-extension}.

For $k\leq\sd(t^0)<\infty$, there are \emph{several} extensions
$t \in \sD'(\bR^k)$ fulfilling the condition $\sd(t) = \sd(t^0)$; the difference of two solutions 
is of the form $\sum_{|\beta|\leq \sd(t^0)-k} C_\beta \,\del^\beta\dl(x)$ with $C_\beta \in \bC$.

The main purpose of the sm-expansion is to reduce perturbative renormalization to
the extension of almost homogeneously scaling distributions.
The following proposition describes the possible
homogeneities of the extensions \cite{DuetschF04,HollandsW02,Hollands08,Hormander90}.

\begin{prop} 
\label{pr:extens-exist}
Let $t^0 \in \sD'(\bR^k \less \{0\})$ scale almost homogeneously with 
degree $D \in \bC$ and power $N_0 \in \bN$ (see \eqref{eq:scalingprop} with
$m\equiv 0$, or \cite[definition 2.4]{DuetschF04}). 
Then there exists an 
extension $t \in \sD'(\bR^k)$ which scales also almost homogeneously
with degree~$D$ and power $N_1 \geq N_0$:
\begin{itemize}
\item[(i)]
if $D \notin \bN_0 + k$, then $t$ is unique and $N_1 = N_0$;
\item[(ii)]
if $D \in \bN_0 + k$, then $t$ is non-unique and $N_1 = N_0$ or
$N_1 = N_0 + 1$. In this case, two solutions differ by a term
$\sum_{|\beta|=D-k} C_\beta \,\del^\beta\dl (x)$ (where $C_\beta \in \bC$ is arbitrary).
\end{itemize}
\end{prop}
In case (i) the unique $t$ can be computed quite easily: if $\Re D< k$
it agrees with the direct extension of $t^0$ \eqref{eq:direct-extension}; otherwise
it can be computed by differential renormalization, see \cite[sect.~4.4]{Keller13}
and sect.~6.

\subsection*{Acknowledgment}
The author profited a lot from stimulating discussions with 
Philippe Blanchard, Klaus Fredenhagen,
Jos{\'e} M. Gracia-Bond{\'i}a, Katarzyna Rejzner and Joseph~C. V\'arilly. 
During working at this paper the author was mainly at the Max Planck Institute 
for Mathematics in the Sciences, Leipzig; he thanks Eberhard Zeidler for the invitations to 
Leipzig and for enlightening discussions.

\end{document}